\documentclass[journal]{IEEEtran}

\usepackage{latexsym,amssymb,exscale,graphicx,amsfonts,amsmath}
\usepackage{multirow,multicol}
\usepackage{epstopdf}
\usepackage{marvosym}
\usepackage{epsfig} 
\usepackage{setspace}
\usepackage[table]{xcolor}
\usepackage{xcolor,colortbl}
\usepackage{lipsum}
\usepackage{microtype}
\usepackage[normalem]{ulem}
\usepackage{stfloats}
\usepackage[noadjust]{cite}
\usepackage{lettrine}
\usepackage{mathtools, cuted}
\usepackage{amsmath}
\usepackage{blkarray}
\usepackage[lined,ruled,linesnumbered,boxed,noend]{algorithm2e}
\usepackage{mathtools}
\usepackage[english]{babel}
\usepackage[utf8]{inputenc}
\usepackage{bbm}
\usepackage{bm}
\usepackage{amsthm, amssymb}
\newtheorem{theorem}{Theorem}[section] 
\newtheorem{corollary}{Corollary}[section]
\newtheorem{lemma}{Lemma}[section]

\newtheorem{assumption}{Assumption}

\usepackage{comment}

\usepackage{amsbsy}
\def\bftheta{\boldsymbol{\theta}}
\def\bfDelta{\boldsymbol{\Delta}}
\def\bfA{\boldsymbol{A}}
\def\bfB{\boldsymbol{B}}
\def\bfP{\boldsymbol{p}}
\def\bfx{\boldsymbol{x}}

\def\bfD{\boldsymbol{D}}
\def\bfDelta{\boldsymbol{\Delta}}

\def\ddefloop#1{\ifx\ddefloop#1\else\ddef{#1}\expandafter\ddefloop\fi}
\def\ddef#1{\expandafter\def\csname bb#1\endcsname{\ensuremath{\mathbb{#1}}}}
\ddefloop ABCDEFGHIJKLMNOPQRSTUVWXYZ\ddefloop
\def\ddef#1{\expandafter\def\csname c#1\endcsname{\ensuremath{\mathcal{#1}}}}
\ddefloop ABCDEFGHIJKLMNOPQRSTUVWXYZ\ddefloop
\def\ddef#1{\expandafter\def\csname v#1\endcsname{\ensuremath{\boldsymbol{#1}}}}
\ddefloop ABCDEFGHIJKLMNOPQRSTUVWXYZabcdefghijklmnopqrstuvwxyz\ddefloop
\def\ddef#1{\expandafter\def\csname v#1\endcsname{\ensuremath{\boldsymbol{\csname #1\endcsname}}}}
\ddefloop {alpha}{beta}{gamma}{delta}{epsilon}{varepsilon}{zeta}{eta}{theta}{vartheta}{iota}{kappa}{lambda}{mu}{nu}{xi}{pi}{varpi}{rho}{varrho}{sigma}{varsigma}{tau}{upsilon}{phi}{varphi}{chi}{psi}{omega}{Gamma}{Delta}{Theta}{Lambda}{Xi}{Pi}{Sigma}{varSigma}{Upsilon}{Phi}{Psi}{Omega}{ell}\ddefloop

\makeatletter
\newcommand{\nosemic}{\renewcommand{\@endalgocfline}{\relax}}
\newcommand{\dosemic}{\renewcommand{\@endalgocfline}{\algocf@endline}}
\let\oldnl\nl
\newcommand{\nonl}{\renewcommand{\nl}{\let\nl\oldnl}}
\makeatother

\newcommand{\ting}[1]{\textcolor{red}{Ting: #1}}
\newcommand{\yudi}[1]{\textcolor{blue}{Yudi: #1}}

\IEEEoverridecommandlockouts

\begin{document}
	
\title{Verifiable Failure Localization in Smart Grid under Cyber-Physical Attacks }

\author{Yudi Huang, \textit{Student Member, IEEE}, Ting He, \textit{Senior Member, IEEE}, Nilanjan Ray Chaudhuri, \\\textit{Senior Member, IEEE}, and Thomas La Porta \textit{Fellow, IEEE}
\thanks{The authors are with the School of Electrical Engineering and Computer Science, Pennsylvania State University, University Park, PA 16802, USA
(e-mail: \{yxh5389, tzh58, nuc88, tfl12\}@psu.edu). }
\thanks{ This work was supported by the National Science Foundation under award ECCS-1836827.}
}

\maketitle

\begin{abstract}
Cyber-physical attacks impose a significant threat to the smart grid, as the cyber attack makes it difficult to {identify the actual} damage caused by the physical attack. To defend against such attacks, various inference-based solutions have been proposed to estimate the states of grid elements (e.g., transmission lines) from measurements outside the attacked area, out of which a few have provided theoretical conditions for guaranteed accuracy. However, these conditions are usually based on the ground truth states and thus not verifiable in practice. To solve this problem, we develop (i) verifiable conditions that can be tested based on only observable information, and (ii) efficient algorithms for verifying the states of links (i.e., transmission lines) within the attacked area based on these conditions. Our numerical evaluations based on the Polish power grid {and IEEE 300-bus system} demonstrate that the proposed algorithms are highly successful in verifying the states of truly failed links, and can thus greatly help in prioritizing repairs during the recovery process. 
%
\end{abstract}
	
\begin{IEEEkeywords}
Power grid state estimation, cyber-physical attack, failure localization, verifiable condition.
\end{IEEEkeywords}

\section{Introduction} \label{sec:Intro}
	
The close interdependency between the physical subsystem (power grid) and its control subsystem (Supervisory Control and Data Acquisition - SCADA or Wide-Area Monitoring Protection and Control - WAMPAC) in modern power grids makes them vulnerable to simultaneous cyber-physical attacks that target both subsystems. Such attacks can cause devastating consequences, e.g., the attack on Ukraine's power grid left 225,000 people without power for days~\cite{UkraineAttack}, 
{as the attacker simultaneously opened circuit breakers (physical attack) while keeping the system operators unaware by jamming the phone lines and launching KillDisk server wiping (cyber attack).} 

As the main challenge in dealing with cyber-physical attacks is the difficulty in accurately identifying the damaged grid elements (e.g., failed transmission lines) due to the lack of measurements (e.g., breaker status) from the attacked area caused by the cyber attack, efforts on countering such attacks have focused on estimating the grid state inside the attacked area using power flow models and measurements outside that area.
Specifically, assuming the post-attack power injections to be known, \cite{Soltan18TCNS} developed methods to estimate the grid state under cyber-physical attacks using the \emph{direct-current (DC) power flow model}, and \cite{Soltan17PES} developed similar methods using the \emph{alternating-current (AC) power flow model}. Recently, \cite{yudi20SmartGridComm} further extends such methods to handle unknown post-attack power injections within the attacked area by proposing a linear programming (LP) based algorithm that can correctly identify all the failed links (i.e., transmission lines) under certain conditions. The conditions, however, involve the ground truth link states {(i.e., whether a link has truly failed or not)} within the attacked area and is thus not verifiable in practice.

In this work, we advance the work in \cite{yudi20SmartGridComm} by developing conditions and algorithms to verify the correctness of link states estimated by the LP-based algorithm proposed therein. 
{Besides providing more confidence in the estimated link states, such algorithms can also facilitate recovery planning after an attack, which will schedule the repairs based on the results of failure localization under resource limitation. As no current algorithm can guarantee $100\%$ localization accuracy and false alarms are costly,
it is highly desirable to verify the correctness of the failure localization results before scheduling repairs. The role of failure localization verification is shown in red in Fig.~\ref{fig:role_verification}, where the set $\hat{F}$ contains the estimated failed links. One application of the proposed method is to guide crew dispatch during line repairing/restoration.
}
\begin{figure}[tb]
\centering
\includegraphics[width=.9\linewidth]{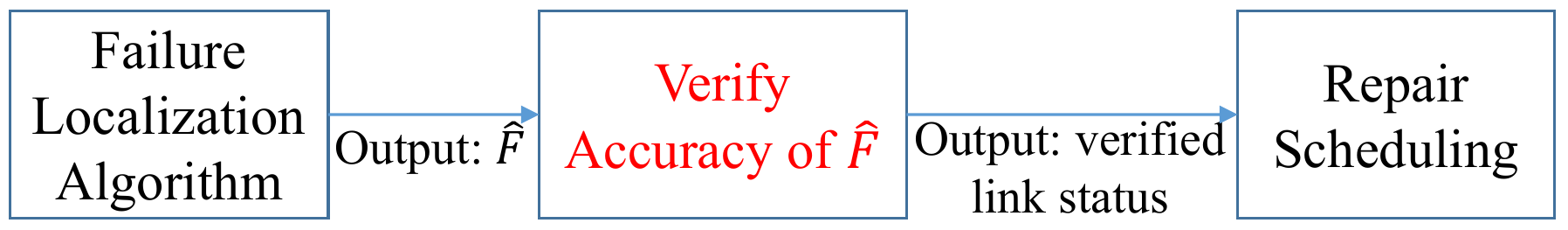}
\vspace{-.5em}
\caption{The role of failure localization verification. 
} \label{fig:role_verification}
\vspace{-1em}
\end{figure}

\subsection{Related Work}


State estimation is of fundamental importance for the supervisory control of the power grid \cite{huang2012state}. Specifically, link status identification or failed link localization is critical for post-attack failure assessment and recovery planning. To detect failed links in physical attacks, early works~\cite{tate2008line, tate2009double} formulated this problem as a mixed-integer program, which cannot scale to multi-link failures. Later works tackled this problem by formulating it as a sparse recovery problem over an overcomplete representation~\cite{Zhu12TPS, chen2014efficient}, which is then relaxed into an LP for computational efficiency, or applying machine learning techniques~\cite{garcia2015line, zhao2019learning}. 

Localizing failed links is more difficult under joint cyber-physical attacks. For cyber attacks that block sensor data to the control center as considered in this work, \cite{Soltan18TCNS} proposed an LP-based algorithm and graph-theoretic conditions for perfect failure localization under the DC power flow model. In \cite{soltan2018react}, a heuristic algorithm was proposed to handle cyber attacks that distort sensor data or inject stealthy data.  Moreover, \cite{soltan2018expose} modified the algorithm and the theoretical guarantees in \cite{Soltan18TCNS} according to the AC power flow model.  However, the above works were all based on the assumption that the power grid remained connected after the failures, which may not be true under multi-link failures~\cite{yudi20SmartGridComm}. Recently, \cite{yudi20SmartGridComm} eliminated this assumption by developing an LP-based algorithm that can jointly estimate the link states and the load shedding values within the attacked area. However, despite the empirical success of this algorithm, there is no existing method for verifying the correctness of its estimates.

{Another line of related works is fault localization, e.g., \cite{Adu01TPD,salim2009extended,Codino17TEC,saha2009fault} and references therein. These works differ from our work in the sense that they (i) target naturally-occurring faults which exhibit signatures not necessarily present during malicious attacks, (ii) mostly focus on finding the exact location of faults along a line as opposed to localizing the failed lines, and (iii) do not traditionally consider the lack of information due to cyber attacks. }
\looseness=-1

\subsection{Summary of Contributions}

We aim at estimating the states (failed/operational) of links (transmission lines) in a smart grid under a joint cyber-physical attack, where the cyber attack blocks sensor data from the attacked area and the physical attack disconnects certain links that may disconnect the grid, with the following contributions: \begin{enumerate}
    \item  
    {We provide conditions and a corresponding algorithm to verify the correctness of failure localization results (the states of links) using only observable information. Compared to previous recovery conditions in~\cite{yudi20SmartGridComm,Huang20arXiv} that cannot be tested during operation, the proposed algorithm requires no information about the ground truth link states and is thus applicable after cyber-physical attacks.}\looseness=-1 

    \item We provide a further theoretical condition for verifying the states of potentially more links based on observable information and the link states that are already verified by the above algorithm, 
    as well as the corresponding verification algorithm.

    \item We show that our conditions and algorithms can be easily adapted to incorporate the knowledge of connectivity if the post-attack grid is known to remain connected as assumed in most existing works.

    \item Our evaluations on 
    {the Polish grid and the IEEE 300-bus system} 
    show that the proposed algorithms can verify {$60$--$80\%$} of failed links and {$40$--$50\%$} of operational links in general, and these numbers increase to {$80$--$95\%$ and $70$--$90\%$} if the post-attack grid is known to remain connected, which provides valuable information for prioritizing repairs during recovery.

\end{enumerate}

\textbf{Roadmap.} Section~\ref{sec:Problem Formulation} formulates our  problem. Section~\ref{sec:Localizing Failed Links} recaps our previously proposed algorithm \cite{yudi20SmartGridComm} and its properties. In Section~\ref{sec: verification_cond}, theoretical conditions and two algorithms are developed to verify the correctness of the estimated link states. Section~\ref{sec:Performance Evaluation} evaluates the proposed algorithms on a real grid topology. Finally, Section~\ref{sec:Conclusion} concludes the paper.

\section{Problem Formulation}\label{sec:Problem Formulation}

\subsection{Power Grid Model}

We adopt the DC power flow model.
The power grid is modeled as a connected undirected graph $G=(V,E)$, where $V$ denotes the set of nodes (buses) and $E$ the set of links (transmission lines). Each link $e=(s,t)$ is associated with a \emph{reactance} $r_{st}$ ($r_{st} = r_{ts}$) and a state $\in \{\mbox{``operational''}, \mbox{``failed''}\}$ (assumed to be operational before attack). Let $\bm{\Gamma}:= \text{diag}\{\frac{1}{r_e}\}_{e\in E}$. Each node $v$ is associated with a phase angle $\theta_v$ and an active power injection $p_v$, which are coupled by DC power flow equation:
\begin{align}\label{eq:B theta = p}
    \bfB \bftheta = \bfP,
\end{align}
where $\bftheta:=(\theta_v)_{v\in V}$, $\bfP:=(p_v)_{v\in V}$, and $\bfB:=(b_{uv})_{u,v\in V} \in \mathbb{R}^{|V|\times|V|}$ is the \emph{admittance matrix}, defined as:
\begin{align}
    b_{uv} &=\left\{\begin{array}{ll}
    0 & \mbox{if }u\neq v, (u,v)\not\in E,\\
    -1/r_{uv} &  \mbox{if } u\neq v, (u,v)\in E,\\
    -\sum_{w\in V\setminus \{u\}}b_{uw} &\mbox{if }u=v.
    \end{array}\right.
\end{align}
Given an arbitrary orientation of the links, the topology of $G$ can also be represented by the \emph{incidence matrix} $\bfD\in \{-1,0,1\}^{|V|\times|E|}$, {where the entry for $u\in V$ and $e\in E$ is\looseness=-1 
\begin{align}
    D_{u,e} &= \left\{\begin{array}{ll}
    1 & \mbox{if link }e\mbox{ comes out of node }u,\\
    -1 & \mbox{if link }e\mbox{ goes into node }u,\\
    0 & \mbox{otherwise.}
    \end{array}\right.
\end{align}
}

We assume that each node is deployed with 
a phasor measurement unit (PMU) that can measure its phase angle, 
and remote terminal units (RTUs) measuring the active power injection, as well as the (breaker) states and the power flows of its incident links. These reports are sent to the control center, where 
the PMU measurements are communicated over a secure WAMPAC network~{\cite{WASA}}, and the RTU measurements over a more vulnerable SCADA network. 

\subsection{Attack Model}


\begin{figure}[tb]
\centering
\includegraphics[width=.6\linewidth]{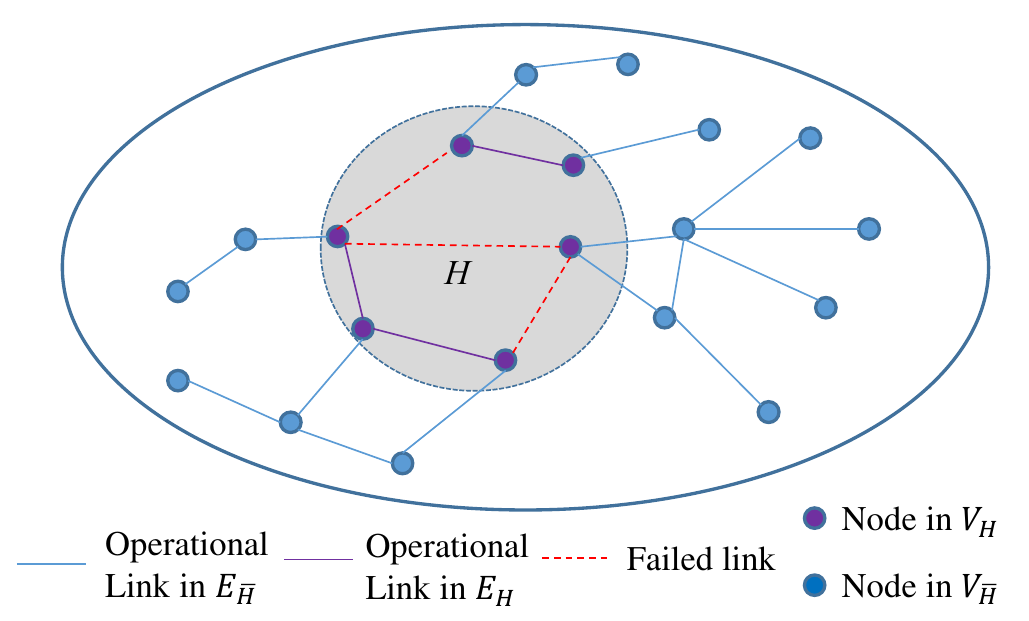}
\vspace{-1em}
\caption{A cyber-physical attack that blocks information from the attacked area $H$ while disconnecting certain links within $H$. 
} \label{fig:cyber_physical_attack_hyd}
\vspace{-.5em}
\end{figure}

As illustrated in Fig.~\ref{fig:cyber_physical_attack_hyd}, a joint cyber-physical attack on an area $H=(V_H,E_H)$ (a subgraph induced by a set of nodes $V_H\subseteq V$) comprises of: (i) cyber attack that blocks reports from the nodes in $V_H$, and (ii) physical attack that disconnects a set $F\subseteq E_H$ of links within $H$, where $E_H$ is the set of links with both endpoints in $V_H$. 
In contrast to the previous works \cite{Soltan18TCNS, Zhu12TPS, chen2014efficient}, we consider that the grid may be decomposed into islands after attack, which leads to possible changes in $\bfP$. Let $\bfDelta = (\Delta_v)_{v\in V} := \bfP-\bfP'$ denote the change in active power injections, where $\bfP'$ denotes the active power injections after the attack.  
Define
\begin{align}\label{eq:tilde{D}}
\tilde{\bfD} := \bfD \bm{\Gamma} \text{diag}\{\bfD^T \bftheta'\},
\end{align}
where $\bftheta'$ denotes post-attack phase angles.
{For link $e=(u,v)$, $\bm{\tilde{D}}_{u,e} =-\bm{\tilde{D}}_{v,e} = \frac{\theta_u'-\theta_v'}{r_{uv}}$ denotes  
the post-attack power flow on $e$ if it is operational. If link $e$ fails after attack, then $\bm{\tilde{D}}_{u,e}$ represents the ``hypothetical power flow''.}

\subsection{Failure Localization Problem}

\begin{table}[tb]
\footnotesize
\renewcommand{\arraystretch}{1.3}
\caption{Notations} \label{tab:notation}
\vspace{-.5em}
\centering
\begin{tabular}{c|l}
  \hline
  Notation & Description  \\
  \hline
 $G=(V,E)$ & power grid \\
 \hline
 $H$, $\bar{H}$ & attacked/unattacked area \\
  \hline
 $F$ & set of failed links \\
 \hline
 $\bfB$ & admittance matrix \\
 \hline
 $\bfD$ & incidence matrix \\
 \hline
 $\bftheta$ & vector of phase angles \\
 \hline
 $\bfP$ & vector of active power injections \\
 \hline
 $\bfDelta$ & vector of changes in active power injections
 \\
 \hline
 $\bm{x}$ & vector of failure indicators \\
  \hline
\end{tabular}
\end{table}
\normalsize

\textbf{Notation.} The main notations are summarized in Table~\ref{tab:notation}.
Moreover, given a subgraph $X$ of $G$, $V_X$ and $E_X$ denote the subsets of nodes/links in $X$, and $\bfx_X$ denotes the subvector of a vector $\bfx$ containing elements corresponding to $X$. Similarly, given two subgraphs $X$ and $Y$ of ${G}$, $\bfA_{X|Y}$ denotes the submatrix of a matrix $\bfA$ containing rows corresponding to $X$ and columns corresponding to $Y$. We use $[A, B]$ to denote the horizontal concatenation of matrices $A, B$ and $\bm{I}_{n}$ to denote the $n\times n$ identity matrix.
We use $\bfD_H\in \{-1,0,1\}^{|V_H|\times|E_H|}$ and $\tilde{\bfD}_H\in \mathbb{R}^{|V_H|\times|E_H|}$ to denote the submatrices of $\bfD$ and $\tilde{\bfD}$ for the attacked area $H$.
For each variable $x$, we use $x'$ to denote its value after the attack. {We follow the convention that $|x|$ indicates the absolute value if $x$ is a scalar and $|A|$ denotes the cardinality if $A$ is a set.}

\textbf{Goal.}
Our goal is to localize the failed links $F$ within the attacked area, based on knowledge before the attack and measurements from the unattacked area $\bar{H}$ after the attack. 
In contrast to \cite{yudi20SmartGridComm}, we aim at obtaining estimates with \emph{verifiable correctness}. \looseness=-1


{\textbf{Assumptions.} Our analysis and solution are based on the following assumptions:\\
1. \emph{DC power flow model:} This is an approximation of the AC power flow model by neglecting resistive losses and assuming a uniform voltage magnitude. Due to its computational efficiency, DC power flow model has been widely used for analyzing link failures in large power grids~\cite{tate2008line,tate2009double,Zhu12TPS,chen2014efficient, Soltan18TCNS,zhao2019learning,soltan2018react}. We leave the extension to the AC power flow model to future work. \\
2. \emph{Availability of phase angles:} We assume that the phase angle at every bus is available before/after the attack. Before-attack observability from PMU measurements is consistent with the goal of PMU deployment, at least in North American bulk transmission systems~\cite{PMUdeployment}. Under the North American SynchroPhasor Initiative (NASPI) \cite{dagle2010north}, the number of PMUs is steadily growing, and some utilities have already achieved full observability in their networks, e.g., Dominion Power has piloted the PMU-based linear state estimator \cite{jones2013three,jones2014methodology}. These trends indicate that it is just a matter of time that complete observability through PMUs is achieved.
%
The post-attack observability can be achieved by securing PMU measurements through the stronger cyber security requirements of WAMPAC~\cite{WAMPACsecurity}, or through inference when $B_{\bar{H}|H}$ has a full column rank~\cite{yudi20SmartGridComm}.  \\
3. \emph{$\theta_s' \ne \theta_t'$ for each link $(s,t)\in E_H$:} This assumption simply means that we only focus on the states of links in $H$ that will carry power flow if not failed, as the states of links carrying no flow have no impact and thus cannot be identified~\cite{Soltan18TCNS, yudi20SmartGridComm}.
}

\section{Estimating Link States}\label{sec:Localizing Failed Links}

To our knowledge, the only algorithm for estimating link states (and hence localizing failed links) under a cyber-physical attack that can disconnect the grid is an algorithm called \emph{Failed Link Detection (FLD)} proposed in \cite{yudi20SmartGridComm}. FLD has exhibited very good accuracy in detecting the failed links with very few false alarms~\cite{yudi20SmartGridComm}. Our idea is to develop algorithms to verify the output of FLD. In this section, we briefly recap FLD and its existing (unverifiable) recovery conditions for completeness. \looseness=-1


\subsection{Existing Algorithm}

Let $\bfx_H \in \{0,1\}^{|E_H| }$ be an indicator vector such that $x_e=1$ if $e\in F$ and $x_e=0$ if $e \in E_H\setminus F$. It has been shown in \cite{yudi20SmartGridComm} that any feasible solution to $\bfx_H$ and $\bfDelta_H$ must satisfy
\begin{align}
&    \bfDelta_H = \bfB_{H|G}(\bftheta-\bftheta') + \bfD_{H}\bm{\Gamma}_H\text{diag}\{\bfD_{G|H}^T \bftheta'\}\vx_H, \label{eq:pf_constraint}\\
& p_v \ge {\Delta_v} \ge 0, ~~~\forall v \in \left\{ {u\: |u \in V_H, p_u > 0} \right\}, \label{eq:const_valid_start}\\
& p_v \le {\Delta_v} \le 0,~~~\forall v \in \left\{ {u\: |u \in V_H, p_u \le 0} \right\}, \label{eq:const_valid_load}
\end{align}
FLD formulates the problem of failure localization as an LP:
\begin{subequations}\label{eq1:L1_binary_load}
\begin{alignat}{2}
(\text{P1}) \quad &\min_{\boldsymbol{x}_H, \bfDelta_H} \Arrowvert \bfx_H \Arrowvert_1 &\\
\mbox{s.t.} \quad
&\eqref{eq:pf_constraint}, \eqref{eq:const_valid_start}, \eqref{eq:const_valid_load},&\\
&{\rm{       }}{\bf{0}} \le {{\bfx}_H} \le {\bf{1}}.&
\end{alignat}
\end{subequations}
which is the convex relaxation of a sparse-recovery-based formulation. After solving (P1) in polynomial time, FLD estimates the set of failed links as 
\begin{align}\label{eq:F_Hat}
    \hat{F}=\{e:x_e \geq \eta\},
\end{align}
where $\eta\in (0, 1)$ is a threshold for rounding the factional solution of $\bfx_H$ to an integral solution ($\eta=0.5$ in this paper). 

\subsection{Existing Recovery Conditions}\label{subsec:Existing Recovery Conditions}

FLD is known to recover the link states correctly under the following conditions~\cite{Huang20arXiv} (which improved the conditions in \cite{yudi20SmartGridComm}), where $\bfx^*_H$ and $\bfDelta^*_H$ denote the true values of $\bfx_H$ and $\bfDelta_H$.  \looseness=-1

\subsubsection{Implicit Conditions}

Denote $V_{L} \subseteq V_H$ as the set containing nodes with $p_v\le 0$, and $V_G := V_H\setminus V_L$ as the remaining nodes in $V_H$ (with $p_v>0$). Accordingly, $\bfDelta_i$ and $\bfP_i$ ($i=L, G$) denote the subvectors of $\bfDelta_H$ and $\bfP_H$, respectively, corresponding to $V_i$, and $\tilde{\bfD}_i$ denotes the submatrix of $\tilde{\bfD}_H$ containing the rows corresponding to $V_i$. Given a set $Q_m := F\setminus \hat{F}$ of failed links that are missed and a set $Q_f := \hat{F}\setminus F$ of operational links that are falsely detected,
define $\bm{W}_m \in \{0,1\}^{|Q_m|\times |E_H|}$ as a binary matrix where  $(W_m)_{i,j} = 1$ indicates the $i$-th missed link to be $e_j$, and  define $\bm{W}_f \in \{0,1\}^{|Q_f|\times |E_H|}$ similarly such that $(W_f)_{i,k} = 1$ if the $i$-th false-alarmed link is $e_k$. Based on these notions, define\looseness=-1
\begin{subequations}\label{eq:inter_sub}
\begin{alignat}{2}
\bm{A}_D^T &:= [\tilde{\bfD}^T_{L}, -\tilde{\bfD}^T_{L}, -\tilde{\bfD}^T_{G}, \tilde{\bfD}^T_{G}] \in \mathbb{R}^{|E_H|\times 2|V_H|}, \\
\bm{A}_x^T &:= [-\bm{I}_{|E_H|}, \bm{I}_{|E_H|}]\in \mathbb{R}^{|E_H|\times 2|E_H|},\\
\bm{W}^T &:= [\bm{W}^T_m, -\bm{W}^T_f]\in \mathbb{R}^{|E_H|\times (|Q_m|+|Q_f|)}, \label{eq:bmW}\\
\bm{g}_D^T&:=[-(\bfDelta_{L}^*)^T, (-\bfP_{L}')^T, (\bfDelta_{G}^*)^T, (\bfP_{G}')^T],\\
\bm{g}_x^T &:= [(\bfx^*_H)^T, \bm{1}^T-(\bfx^*_H)^T]\in \mathbb{R}^{1\times 2|E_H|},\\
\bm{g}_w^T &:= [(\eta-1)\bm{1}^T, -\eta\bm{1}^T]\in \mathbb{R}^{1\times (|Q_m|+|Q_f|)}.
\end{alignat}
\end{subequations}
Then, the correctness of FLD is guaranteed as follows.

\begin{lemma}[\cite{Huang20arXiv}]\label{lem:ground_alter_gale}
A link $e\in F$ cannot be missed ($e\notin F \setminus \hat{F}$) by FLD  if for any $Q_m$ containing $e$, there is a solution $\bm{z}\ge \bm{0}$ to \looseness=-1
\begin{subequations}\label{eq:alter_gale}
\begin{alignat}{2}
[\bm{A}_D^T, \bm{A}_x^T, \bm{W}^T, \bm{1}]\bm{z} = \bm{0},\label{eq:alter_gale_eq} \\
[\bm{g}_D^T, \bm{g}_x^T, \bm{g}_w^T, \bm{0}]\bm{z} < 0.\label{eq:alter_gale_ineq}
\end{alignat}
\end{subequations}
Similarly, a link $e\in E_H \setminus F$ cannot be falsely detected as failed if
for any $Q_f$ with $e \in Q_f$, 
there is a solution $\bm{z}\ge\bm{0}$ to \eqref{eq:alter_gale}.\looseness=-1 
\end{lemma}

\subsubsection{Explicit Conditions}

Besides Lemma~\ref{lem:ground_alter_gale}, \cite{Huang20arXiv} also provided more explicit conditions in terms of post-attack power flows and power injections. The following definitions will be needed to present this result.
Let $\bm{z}_D\in \mathbb{R}^{2|V_{H}|}, \bm{z}_x\in \mathbb{R}^{2|E_{H}|}$, $\bm{z}_w \in \mathbb{R}^{|Q_m|+|Q_f|}$ and $z_*\in \mathbb{R}$ denote subvectors of $\bm{z}$ corresponding to $\bm{A}_D^T, \bm{A}_x^T, \bm{W}^T$, and $\bm{1}$ in \eqref{eq:alter_gale_eq}. Denote $\tilde{\bfD}_{u}$ as the row in $\tilde{\bfD}$ corresponding to node $u$, and $\tilde{D}_{u,e}$ as the entry in $\tilde{\bfD}_{u}$ corresponding to link $e$. Denote $z_{D,u}$ as the entry in $\bm{z}_D$ corresponding to $\tilde{\bfD}_{u}$ in $\bm{A}_D$ and $z_{D,-u}$ as the entry corresponding to  $-\tilde{\bfD}_{u}$ in $\bm{A}_D$.
Define $g_{D,u}$ and $g_{D,-u}$ as the entries in $\bm{g}_D$ corresponding to $z_{D,u}$ and $z_{D,-u}$, respectively, i.e.,\looseness=-1
\begin{align}
&g_{D,u}\hspace{-.25em} := \hspace{-.25em}\left\{\hspace{-.25em}\begin{array}{ll}
-\Delta_u^* & \hspace{-.5em}\mbox{if } p_u \hspace{-.25em}\le\hspace{-.25em} 0,\\
p_u' & \hspace{-.5em}\mbox{if } p_u\hspace{-.25em}>\hspace{-.25em} 0,
\end{array}\right.
&\hspace{-1em} g_{D,-u} \hspace{-.25em}:=\hspace{-.25em}  \left\{\hspace{-.25em}\begin{array}{ll}
-p_u' & \hspace{-.5em}\mbox{if } p_u \hspace{-.25em}\le\hspace{-.25em} 0,\\
\Delta_u^* & \hspace{-.5em}\mbox{if } p_u \hspace{-.25em}>\hspace{-.25em} 0.
\end{array}\right.
\end{align}
Moreover, if link $e$ is the $i^{th}$ link in $Q_m$, then $z_{w,m,e}$ is used to denote the entry in $\bm{z}_{w}$ that corresponds to the $i^{th}$ column of $\bm{W}_m^T$; $z_{w,f,e}$ is defined similarly if $e\in Q_f$. For each link $e$, we denote $z_{x-,e}$ as the entry in $\bm{z}_x$ corresponding to $x^*_e$ in $\bm{g}_x$ and $z_{x+,e}$ as the entry corresponding to $(1-x^*_e)$ in $\bm{g}_x$.\looseness=-1

Referring to a set of nodes $U \subseteq V_H$ that induce a connected subgraph before attack as a \emph{hyper-node}, 
%
\cite{Huang20arXiv} established recovery conditions based on the following attributes of hyper-nodes. Define $E_U$ as the set of links in $H$ with exactly one endpoint in $U$, i.e, $E_U := \{e|e=(s,t)\in E_H, s\in U, t\notin U\}$. If $E_U\cap F \neq \emptyset$, define:
\begin{subequations}\label{eq:properties of hyper-node}
\begin{align}
\tilde{D}_{U,e} &:= \sum_{u\in U}\tilde{D}_{u,e},\\
S_U &:= \{ e\in E_U\setminus F|\: \exists l \in E_U \cap F,  \tilde{D}_{U,l}\tilde{D}_{U,e} > 0\}, \\
f_{U,g} &:= \hspace{-.25em}
\begin{cases}
{\sum_{u\in U}{g_{D,u}}\mbox{\ \ \ if } {\exists l \in {E_U} \cap F, {\tilde D}_{U,l} < 0} },\\
\sum_{u\in U}{g_{D, - u}} \mbox{\ otherwise.}\\
\end{cases}\label{eq:def_fug}
\end{align}
\end{subequations}
{An illustrative example of hyper-node is $U=\{u_1, u_2, u_3\}$ in Fig.~\ref{fig:eg_hyper_node}, where $E_U=\{l_2, l_4, l_6, l_7\}$.} If $E_U\cap F = \emptyset$, we define:
\begin{align}
{f_{U,g}} := \left\{ {\begin{array}{*{20}{ll}}
{\sum_{u\in U}{g_{D,u}}\mbox{\ \ \ if } {\exists l \in {E_U} \setminus F, {\tilde D}_{U,l} > 0} },\\
\sum_{u\in U}{g_{D, - u}} \mbox{\ otherwise.} 
\end{array}} \right.
\end{align}


\begin{theorem}[\cite{Huang20arXiv}]\label{lem:no_Miss_hyper_node}
A failed link $l \in F$ will be detected by FLD, i.e., $l\in \hat{F}$, if there exists at least one hyper-node (say $U$) such that $l\in E_U$, for which the following conditions hold:\looseness=-1
\begin{enumerate}
    \item $\forall e, l\in E_U\cap F$, $\tilde{D}_{U,e}\tilde{D}_{U,l} > 0$,
    \item $S_U = \emptyset$, and
    \item $f_{U,g} + (\eta-1)|\tilde{D}_{U,l}|<0$.
\end{enumerate}
\end{theorem}

\begin{theorem}[\cite{Huang20arXiv}]\label{lem:no_fa_hyper_direc}
An operational link $l \in E_H\setminus F$ will not be detected as failed by FLD, i.e., $l\notin \hat{F}$, if there exists at least one hyper-node (say $U$) such that $l\in E_U$, for which the following conditions hold:\looseness=0
\begin{enumerate}
    \item 
    $\forall l, l'\in E_U\setminus F,\: \tilde{D}_{U,l}\tilde{D}_{U,l'} > 0$,
    \item $S_U = \emptyset$ if $E_U\cap F \ne \emptyset$, and
    \item $f_{U,g}-\eta|\tilde{D}_{U,l}|<0$.
\end{enumerate}
\end{theorem}

While useful for performance analysis, the above conditions cannot be directly applied to verify whether the estimated state of a link is correct or not as the ground truth $F$ is unknown.

\section{Verifying Estimated Link States}\label{sec: verification_cond}

We will show that in some cases, we can guarantee the correctness of estimated link states based on observable information. Our idea is to (1) derive stronger recovery conditions that can be tested without knowledge of the ground truth {link states}, and then (2) extend these conditions to test more links based on the link states verified in step~(1).

Our results are based on the assumption that the grid follows the \emph{proportional load shedding/generation reduction policy}, where (i) either the load or the generation (but not both) will be reduced upon the formation of an island, and (ii) if nodes $u$ and $v$ are in the same island and of the same type (both load or generator), then $p_u'/p_u = p_v'/p_v$. This policy models the common practice in adjusting load/generation due to islanding~\cite{pal2006robust,lu2016under}.  Under this policy,  it is known that the post-attack power injections can be recovered under the following condition.

\begin{lemma}[\cite{yudi20SmartGridComm}]\label{lem:recover delta}
Let $N(v;\bar{H})$ denote the set of all the nodes in $\bar{H}$ that are connected to node $v$ via links in $E\setminus E_H$.
Then under the proportional load shedding/generation reduction policy, $\Delta_v$ for $v\in V_H$ can be  recovered unless $N(v;\bar{H})=\emptyset$ or every $u\in N(v;\bar{H})$ is of a different type from $v$ with $\Delta_u=0$.
\end{lemma}
Define $U_B$ as the set of nodes such that $\forall u\in U_B$, $\Delta_u$ can be recovered through Lemma~\ref{lem:recover delta}.

Our key observation is that for any hyper-node $U$, $\tilde{D}_{U,l}$ for any $l\in E_U$ can be computed with the knowledge of $\bftheta'$, and $f_{U,g}$ can be upper-bounded by
\begin{align}
\hat{f}_{U,g} := \sum_{u\in U \cap U_B} f_{u,g} + \sum_{u\in U\setminus U_B} |p_u|,
\end{align}
where $f_{u,g}$ is defined in \eqref{eq:def_fug} for $U=\{u\}$. Since $f_{u,g}$ is known for nodes in $U_B$ and $p_u$ (power injection at $u$ before attack) is also known, $\hat{f}_{U,g}$ is computable. We now show how to use this information to verify the estimated link states based on Lemma~\ref{lem:ground_alter_gale} and  Theorems~\ref{lem:no_Miss_hyper_node}--\ref{lem:no_fa_hyper_direc}.

\subsection{Verification without Knowledge of Ground Truth}

We first tackle the links whose states can be verified without any knowledge of the ground truth {link states}.

\subsubsection{Verifiable Conditions}\label{subsubsec:Verifiable Condition, no ground truth}

The basic idea is to rule out the other possibility by constructing \emph{counterexamples} to the theorems in Section~\ref{subsec:Existing Recovery Conditions} {if the estimated link state is incorrect.}

\begin{figure}[tb]
\vspace{-.5em}
\centering
\includegraphics[width=.6\linewidth]{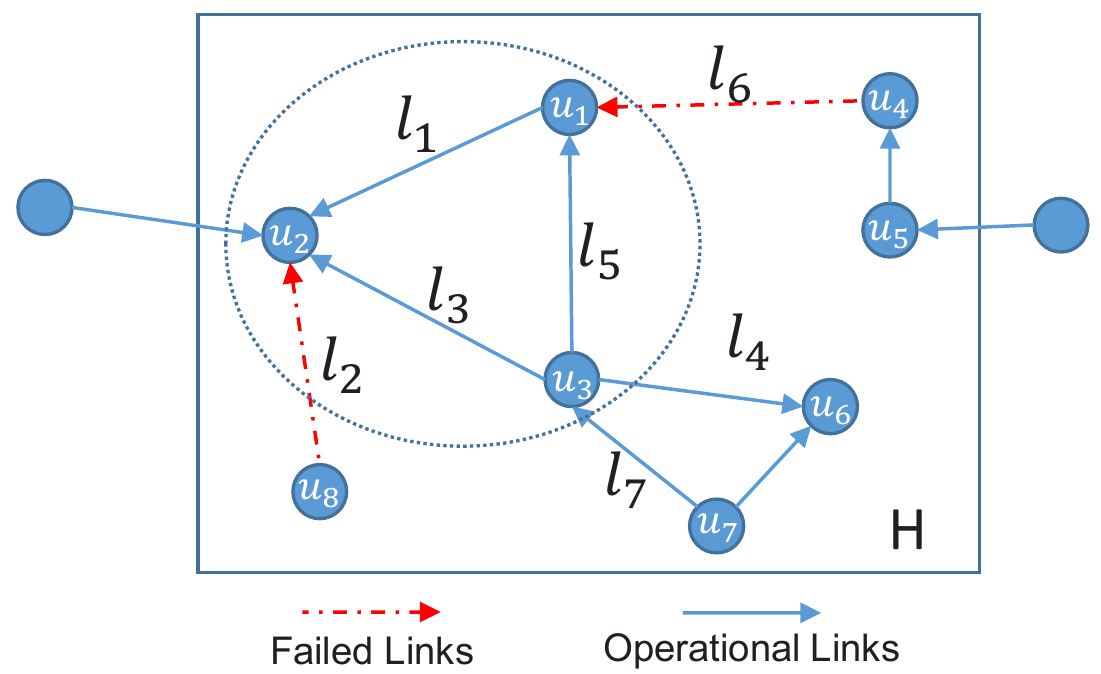}
\vspace{-1em}
\caption{An example of hyper-node (arrow denotes the direction of a power flow over an operational link or a hypothetical power flow over a failed link). } \label{fig:eg_hyper_node}
\vspace{-1em}
\end{figure}

\emph{Links in $1$-edge cuts:}
If link $e = (u_1,u_2)$ forms a \emph{cut} of $H$, i.e., $(V_H, E_H\setminus \{e\})$ contains more connected components than $H$, then by breadth-first search (BFS) starting from $u_1$ and $u_2$ respectively without traversing $e$, we can construct two hyper-nodes $U_1$ and $U_2$ such that $E_{U_1} = E_{U_2} = \{e\}$ and thus $S_{U_1}=S_{U_1}=\emptyset$. {For example, in Fig.~\ref{fig:eg_hyper_node}, link $e:=l_6$ is a 1-edge cut, and thus $U_1:=\{u_4,u_5\}$ and $U_2:=V_H\setminus U_1$ satisfy this condition. }
Then the following verifiable conditions are directly implied by Theorems~\ref{lem:no_Miss_hyper_node}--\ref{lem:no_fa_hyper_direc}:

\begin{corollary}\label{coro:verifiable condition of 1-edge cut}
If $e\in \hat{F}$ and $\min\{ \hat{f}_{U_1,g}, \hat{f}_{U_2,g} \} - \eta|\tilde{D}_{U_1,e}| <0$, then we can verify $e\in F$.
If $e\in E_H\setminus \hat{F}$ and $\min\{ \hat{f}_{U_1,g}, \hat{f}_{U_2,g} \} + (\eta-1)|\tilde{D}_{U_1,e}| <0$, then we can verify $e\in E_H\setminus F$.
\end{corollary}
\begin{proof}
If $e\in \hat{F}$ and $\min\{ \hat{f}_{U_1,g}, \hat{f}_{U_2,g} \} - \eta|\tilde{D}_{U_1,e}| <0$, then $e$ must have failed, since otherwise $e$ would have been estimated as operational according to Theorem~\ref{lem:no_fa_hyper_direc}. Similarly, if $e\in E_H\setminus \hat{F}$ and $\min\{ \hat{f}_{U_1,g}, \hat{f}_{U_2,g} \} + (\eta-1)|\tilde{D}_{U_1,e}| <0$, then $e$ must be operational, since otherwise $e$ would have been estimated as failed according to Theorem~\ref{lem:no_Miss_hyper_node}.
Note that 
as our verification is based on contradiction, $\hat{f}_{U_i,g}$ should be computed as
if $e\in E_H\setminus F$ to verify $e\in \hat{F}$ and vice-versa.
\end{proof}

\emph{Links in $2$-edge cuts:}
If links $e_1, e_2\in E_H$ together form a cut of $H$ but {each individual link does} not, 
then by BFS starting from the endpoints of $e_1$ (or $e_2$) without traversing $e_1$ or $e_2$, we can construct two hyper-nodes $U_1, U_2$ such that $E_{U_1} = E_{U_2} = \{e_1, e_2\}$. {For example, as $e_1:=l_4$ and $e_2:=l_7$ form a $2$-edge cut of $H$ in Fig.~\ref{fig:eg_hyper_node},  $U_1:=\{u_6,u_7\}$ and $U_2:=V_H\setminus U_1$ satisfy this condition.  Moreover, any pair of links in a cycle $C$ form a 2-edge cut if they are not in any other cycle in $H$, e.g., any pair of links in the cycle $\{l_1, l_3, l_5\}$ satisfy this condition.} Based on this observation, we provide the following conditions for verifying the states of such links. 

\begin{theorem}\label{lem:certify_Ef_0}
Consider a hyper-node $U$ with $E_U = \{e_1,e_2\}$ and $e_1, e_2 \in  E_H\setminus \hat{F}$. If $\tilde{D}_{U,e_1}\tilde{D}_{U,e_2} < 0$, then $e_1, e_2$ are guaranteed to both belong to $E_H\setminus F$ if
\begin{enumerate}
    \item $\hat{f}_{U,g} + (\eta-1) \min\{|\tilde{D}_{U,e_1}|, |\tilde{D}_{U,e_2}|\}  < 0$, and
    \item $\eta < 1- \min\{ \frac{\hat{f}_{U,g} + |\tilde{D}_{U,e_1}|}{|\tilde{D}_{U,e_2}|}, \frac{\hat{f}_{U,g} + |\tilde{D}_{U,e_2}|}{|\tilde{D}_{U,e_1}|} \} $.
\end{enumerate}
If $\tilde{D}_{U,e_1}\tilde{D}_{U,e_2} > 0$, then we can verify: 
\begin{enumerate}
    \item $e_1\in E_H\setminus F$ if $(1-\eta)|\tilde{D}_{U,e_1}| > \hat{f}_{U,g} + |\tilde{D}_{U,e_2}|$,
    \item $e_2 \in E_H\setminus F$ if $(1-\eta)|\tilde{D}_{U,e_2}| > \hat{f}_{U,g} + |\tilde{D}_{U,e_1}|$.
\end{enumerate}
\end{theorem}
\begin{proof}
We first prove the case that $\tilde{D}_{U,e_1}\tilde{D}_{U,e_2} < 0$. Given $e_1, e_2 \in  E_H\setminus \hat{F}$ where $\hat{F}$ is returned by FLD, there are 3 possible forms of mistakes when the ground truth {failed link set $F$} is unknown, and we will prove the impossibility for each of them. If $e_1\in F, e_2\in E_H\setminus F$, Theorem~\ref{lem:no_Miss_hyper_node} guarantees that $e_1\notin Q_m$ due to condition~1), which {introduces} contradiction. Similarly, $e_2\in F, e_1\in E_H\setminus F$ is also impossible. If $e_1, e_2 \in Q_m$, assume without loss of generality that $\eta <1- \frac{\hat{f}_{U,g} + |\tilde{D}_{U,e_1}|}{|\tilde{D}_{U,e_2}|}$. Then, we construct the following $\bm{z}$: $\forall u\in U$, $z_{D,u} = 1$ if $\tilde{D}_{U,e_2} < 0$ or $z_{D,-u} = 1$ if $\tilde{D}_{U,e_2} > 0$, $z_{w,m,e_2} = |\tilde{D}_{U,e_2}|$, $z_{x-,e_1} = |\tilde{D}_{U,e_1}|$, and other entries of $\bm{z}$ as 0. Then, \eqref{eq:alter_gale_eq} holds for sure and \eqref{eq:alter_gale_ineq} holds since it can be expanded as $\hat{f}_{U,g} + (\eta-1)|\tilde{D}_{U,e_2}| + |\tilde{D}_{U,e_1}| < 0$ due to condition~2). According to Lemma~\ref{lem:ground_alter_gale}, {it is impossible to have} $e_1, e_2 \in Q_m$, which verifies {that} $e_1, e_2 \in  E_H\setminus F$.

Next, with $\tilde{D}_{U,e_1}\tilde{D}_{U,e_2} > 0$, we show how to verify $e_1$. If $e_1\in Q_m$, regardless of the true {state} of $e_2$, we construct the following $\bm{z}$ for Lemma~\ref{lem:ground_alter_gale}: $\forall u\in U$, $z_{D,u} = 1$ if $\tilde{D}_{U,e_1} < 0$ or $z_{D,-u} = 1$ if $\tilde{D}_{U,e_1} > 0$, $z_{w,m,e_1} = |\tilde{D}_{U,e_1}|$, $z_{x+,e_2} = |\tilde{D}_{U,e_2}|$, and other entries of $\bm{z}$ as 0. Then \eqref{eq:alter_gale} holds due to condition~1), which contradicts {the} assumption that $e_1\in Q_m$. The verification condition for $e_2$ can be derived similarly.
\end{proof}

\begin{theorem}\label{lem:certify_Ef_1}
Consider a hyper-node $U$ with $E_U = \{e_1,e_2\}$ and $e_1 \in \hat{F},e_2\in E_H\setminus \hat{F}$. If $\tilde{D}_{U,e_1}\tilde{D}_{U,e_2} > 0$, then the states of $e_1, e_2$ are guaranteed to be correctly identified if
\begin{enumerate}
    \item $\hat{f}_{U,g} - \eta|\tilde{D}_{U,e_1}| < 0$, $\hat{f}_{U,g} + (\eta-1) |\tilde{D}_{U,e_2}| < 0$, and
    \item either $\eta > \frac{\hat{f}_{U,g} + |\tilde{D}_{U,e_2}|}{|\tilde{D}_{U,e_1}|}$ or $\eta < 1-\frac{\hat{f}_{U,g} + |\tilde{D}_{U,e_1}|}{|\tilde{D}_{U,e_2}|}$.
\end{enumerate}
If $\tilde{D}_{U,e_1}\tilde{D}_{U,e_2} < 0$, then we can verify: 
\begin{enumerate}
    \item $e_1\in F$ if $\eta|\tilde{D}_{U,e_1}| > \hat{f}_{U,g} + |\tilde{D}_{U,e_2}|$,
    \item $e_2 \in E_H\setminus F$ if $(1-\eta)|\tilde{D}_{U,e_2}| >\hat{f}_{U,g} + |\tilde{D}_{U,e_1}|$.
\end{enumerate}
\end{theorem}
\begin{proof}
We first prove the impossibility of each possible mistake if $\tilde{D}_{U,e_1}\tilde{D}_{U,e_2} > 0$. First, we rule out the possibility that $e_1 \in Q_f$, $e_2 \in E_H\setminus F$ according to Theorem~\ref{lem:no_fa_hyper_direc} and condition~1). Similarly, according to Theorem~\ref{lem:no_Miss_hyper_node} and condition~1), $e_1\in F$ while $e_2\in Q_m$ is also impossible. Next, we prove the impossibility of $e_1\in Q_f, e_2\in Q_m$ by constructing a solution $\bm{z}$ to \eqref{eq:alter_gale}. Specifically, if $\eta > \frac{\hat{f}_{U,g} + |\tilde{D}_{U,e_2}|}{|\tilde{D}_{U,e_1}|}$, then $\forall u\in U$, we set $z_{D,u} = 1$ if $\tilde{D}_{U,e_1} > 0$ or $z_{D,-u} = 1$ if $\tilde{D}_{U,e_1} < 0$, $z_{w,f,e_1} = |\tilde{D}_{U,e_1}|$, $z_{x-,e_2} = |\tilde{D}_{U,e_2}|$, and other entries of $\bm{z}$ as 0. If $\eta < 1-\frac{\hat{f}_{U,g} + |\tilde{D}_{U,e_1}|}{|\tilde{D}_{U,e_2}|}$, then $\forall u\in U$, we set $z_{D,u} = 1$ if $\tilde{D}_{U,e_2} < 0$ or $z_{D,-u} = 1$ if $\tilde{D}_{U,e_2} > 0$, $z_{w,m,e_2} = |\tilde{D}_{U,e_2}|$, $z_{x+,e_1} = |\tilde{D}_{U,e_1}|$, and other entries of $\bm{z}$ as 0. It is easy to check the satisfaction of \eqref{eq:alter_gale} under both constructions above, which {rules} out the possibility of $e_1\in Q_f, e_2\in Q_m$ according to Lemma~\ref{lem:ground_alter_gale} and $e_1 \in F,e_2\in E_H\setminus F$ is thus guaranteed.

Next, we prove the verification condition for $e_1\notin Q_f$ if $\tilde{D}_{U,e_1}\tilde{D}_{U,e_2} < 0$. We prove by constructing a solution $\bm{z}$ as follows regardless of the status of $e_2$: $\forall u\in U$, if $\tilde{D}_{U,e_1} < 0$, we set $z_{D,-u} = 1$; otherwise, we set $z_{D,u} = 1$. Then, we set $z_{w,f,e_1} = |\tilde{D}_{U,e_1}|$, $z_{x+,e_2} = |\tilde{D}_{U,e_2}|$, and other entries of $\bm{z}$ as 0.  Then, \eqref{eq:alter_gale_eq} holds for sure and \eqref{eq:alter_gale_ineq} holds since it can be expanded as $\hat{f}_{U,g} -\eta|\tilde{D}_{U,e_1}| + |\tilde{D}_{U,e_2}| < 0$ due to condition 1), which rules out the possibility of $e_1\in Q_f$ according to Lemma~\ref{lem:ground_alter_gale} and thus verifies that $e_1\in F$. The verification condition for $e_2\notin Q_m$ can be proved similarly.
\end{proof}

\begin{theorem}\label{lem:certify_Ef_2}
Consider a hyper-node $U$ with $E_U = \{e_1,e_2\}$ and $e_1, e_2 \in \hat{F}$. Then, we can verify: 
\begin{enumerate}
    \item $e_1\in F$ if $\eta|\tilde{D}_{U,e_1}| > \hat{f}_{U,g} + |\tilde{D}_{U,e_2}|$,
    \item $e_2 \in F$ if $\eta|\tilde{D}_{U,e_2}| > \hat{f}_{U,g} + |\tilde{D}_{U,e_1}|$.
\end{enumerate}
\end{theorem}
\begin{proof}
We only prove the verification condition for $e_1\in F$ since the condition for $e_2$ can be proved similarly. We prove by contradiction that constructs a solution to \eqref{eq:alter_gale} if $e_1\in Q_f$. Specifically, with condition~1), we can always construct {a} $\bm{z}$ for \eqref{eq:alter_gale} as follows regardless of the status of $e_2$: $\forall u\in U$, $z_{D,u} = 1$ if $\tilde{D}_{U,e_1} > 0$ or $z_{D,-u} = 1$ if $\tilde{D}_{U,e_1} < 0$ and $z_{w,f,e_1} = |\tilde{D}_{U,e_1}|$. In addition, if $\tilde{D}_{U,e_1}\tilde{D}_{U,e_2}>0$, we set $z_{x-,e_2} = |\tilde{D}_{U,e_2}|$; otherwise, we set $z_{x+,e_2} = |\tilde{D}_{U,e_2}|$. Finally, other entries of $\bm{z}$ are set as 0. It is easy to check the satisfaction of \eqref{eq:alter_gale_eq}, and \eqref{eq:alter_gale_ineq} holds since it can be expanded as $[\bm{g}_D^T, \bm{g}_x^T, \bm{g}_w^T, \bm{0}]\bm{z} \le \hat{f}_{U,g} + |\tilde{D}_{U,e_2}|-\eta |\tilde{D}_{U,e_1}| < 0$, where {the} last inequality holds due to condition~1). Thus, {we must have $e_1\notin Q_f$} according to Lemma~\ref{lem:ground_alter_gale}.
\end{proof}
\vspace{-0.6em}
\emph{Remark:} While in theory such verifiable conditions can also be derived for links in larger cuts, the number of cases will grow exponentially. We also find $1$--$2$-edge cuts to cover the majority of links in practice (see Fig.~\ref{fig:verifytopology_H40}).

\subsubsection{Verification Algorithm}

Based on Lemmas~\ref{lem:certify_Ef_0}--\ref{lem:certify_Ef_2}, we develop an algorithm as shown in Algorithm~\ref{alg: Verification_Alg} for verifying the link states estimated by FLD, which can be applied to links in $1$--$2$-edge cuts. 
Here, $E_a$ denotes the set of all the links in $1$-edge cuts of $H$, while $\mathcal{E}_c$ denotes the set of $2$-edge cuts. In the algorithm, links in $E_a$ are tested before links in $\mathcal{E}_c$ since it is easier to extend the knowledge of $U_B$ based on the test results for $E_a$.
As for the complexity, {we} first note that the time complexity of each iteration is $\mathcal{O}(|E_H|+|V_H|)$  due to BFS. Then, it takes $\mathcal{O}(|E_H|)$ iterations to verify $E_a$ and $\mathcal{O}(|E_H|^2)$ iterations for $\mathcal{E}_c$, which results in a total complexity of $\mathcal{O}(|E_H|^2(|E_H|+|V_H|))$.

\begin{algorithm}\label{alg: Verification_Alg}
\SetAlgoLined
\SetKwFunction{Fmain}{FailEdgeDetection}
\SetKwInOut{Input}{input}\SetKwInOut{Output}{output}
\KwIn{$\tilde{\bfD}, \bfP, \bfDelta_{\bar{H}}, U_B, \eta, E_a, \mathcal{E}_c, \hat{F}$ 
}
\KwOut{$E_v$}
$E_v\leftarrow \emptyset$\tcc*{verifiable links}
\ForEach{$e = (u_1,u_2) \in E_a$}{
    Construct hyper-nodes $U_1$ and $U_2$ such that $E_{U_1} = E_{U_2} = \{e\}$\;
    \eIf{$e\in \hat{F}$}{
    Add $e$ to $E_v$ if $\min\{ \hat{f}_{U_1,g}, \hat{f}_{U_2,g} \} - \eta|\tilde{D}_{U_1,e}| <0$\;
    }
    {
    Add $e$ to $E_v$ if $\min\{ \hat{f}_{U_1,g}, \hat{f}_{U_2,g} \} + (\eta-1)|\tilde{D}_{U_1,e}| <0$\;
    }
    \If{$e$ is verified to be in $E_H \setminus F$}{
    Add $u_i$ to $U_B$ if $\Delta_{u_i}$ ($i=1,2$) can be recovered through Lemma~\ref{lem:recover delta}\;
    }
}
    \ForEach{$\{e_1, e_2\}\in \mathcal{E}_c$}{
        Construct hyper-nodes $U_1$ and $U_2$ such that $E_{U_1} = E_{U_2} = \{e_1, e_2\}$\;
        Test the satisfaction of Lemma~\ref{lem:certify_Ef_0}, \ref{lem:certify_Ef_1}, or \ref{lem:certify_Ef_2} for $U_1$ and $U_2$, respectively\;
        Add $e_i$ ($i=1,2$) to $E_v$ if it is verified\;
    }
\caption{Verification without Ground Truth}
\end{algorithm}

\subsection{Verification with Partial Knowledge of Ground Truth}

Algorithm~\ref{alg: Verification_Alg} assumes no knowledge of the ground truth {link states}, even if the states of some links are already verified. However, links that cannot be verified by Algorithm~\ref{alg: Verification_Alg} may become verifiable after obtaining partial knowledge of the ground truth (i.e., link set $E_v$ verified by Algorithm~\ref{alg: Verification_Alg}). In addition, links in larger cuts are not tested in Algorithm~\ref{alg: Verification_Alg}. To address these issues, we propose a followup step designed to verify the states of the links in $E_H\setminus E_v$.

\subsubsection{Verifiable Conditions}

The idea for verifying the correctness of $e\in \hat{F}$ (or $e\in E_H\setminus \hat{F}$) is to construct a solution to \eqref{eq:alter_gale} as if $e\in E_H\setminus F$ (or $e\in F$). Specifically, it can be shown that for a link $e\in \hat{F}$, if there exists $\bm{z}\geq \bm{0}$ for \eqref{eq:alter_gale} where $\bm{W}$ is constructed for $Q_f = \{e\}$ and $Q_m=\emptyset$, then $e$ is guaranteed to have failed since otherwise it must have been estimated to be operational. The challenge is the unknown $\bm{g}_D$, $\bm{g}_x$, and $\bm{g}_w$ due to unknown $F$ and $\bfDelta_H^*$. To tackle this challenge, we approximate these parameters by their worst possible values
(in terms of satisfying \eqref{eq:alter_gale}), which leads to the following result:

\begin{theorem}\label{theo: verify_by_LP}
Given a set $E_v$ of links with known states, we define $\hat{\bm{g}}_D \in \mathbb{R}^{2|V_H|}$ and $\hat{\bm{g}}_x\in \mathbb{R}^{2|E_H|}$ as follows:
\begin{align}
\hat{g}_{D,u} &=
\begin{cases}
{g_{D,u}}, \mbox{ if $u\in U_B$,} \\
\left| {{p_u}} \right|, \mbox{ otherwise, }
\end{cases}
~~\hat{g}_{x,e} &=
\begin{cases}
{g_{x,e}}, \mbox{ if $e\in E_v$, }\\
1, \mbox{   otherwise, }
\end{cases}\nonumber
\end{align}
and define $\hat{g}_{D,-u}$ and $\hat{g}_{x,-e}$ similarly. Then, a link $l\in \hat{F}$ is verified to have failed if there exists a solution $\bm{z} \ge \bm{0}$ to
\begin{subequations}\label{eq:computable_gale}
\begin{alignat}{2}
[\bm{A}_D^T, \bm{A}_x^T, \bm{w}^T, \bm{1}]\bm{z} = \bm{0},\label{eq: computable_gale_eq} \\
[\hat{\bm{g}}_D^T, \hat{\bm{g}}_x^T, g_w, \bm{0}]\bm{z} < 0,\label{eq: computable_gale_ineq}
\end{alignat}
\end{subequations}
where $\bm{w} \in \{0,1\}^{|E_H|}$ is defined to be $\bm{W}_f$ with $Q_f = \{l\}$, and $g_w := -\eta$. Similarly, a link $e \in E_H \setminus \hat{F}$ is verified to be operational if $\exists \bm{z} \ge \bm{0}$ that satisfies \eqref{eq:computable_gale}, where $\bm{w}\in \{0,1\}^{|E_H|}$ is defined to be $\bm{W}_m$ with $Q_m = \{e\}$, and $g_w := \eta -1$.
\end{theorem}
\begin{proof}
We only prove for the case that $l \in \hat{F}$ since the case that $e\in E_H\setminus \hat{F}$ is similar. {First note that if $\exists \bm{z}_0\geq \bm{0}$ that satisfies \eqref{eq:alter_gale} for $\bm{W}$ constructed according to $Q_f = \{l\}$ and $Q_m = \emptyset$, then for any $\bm{W}$ corresponding to $Q_f$ that contains $l$, we can always construct a non-negative solution to \eqref{eq:alter_gale} based on $\bm{z}_0$ by setting $z_{w,f,e'} = 0, \forall e'\in Q_f\setminus \{l\}$.} Thus, according to Lemma~\ref{lem:ground_alter_gale}, $l$ {can be verified} as $l \in F$ if $\exists \bm{z}\ge \bm{0}$ for \eqref{eq:alter_gale} where $\bm{W}$ is constructed for $Q_f = \{l\}$ and $Q_m = \emptyset$, since otherwise $l$ must have been estimated to be operational.
%
%
Thus, we only need to prove that any solution to \eqref{eq:computable_gale} is a solution to \eqref{eq:alter_gale} when $Q_f=\{l\}$ and $Q_m=\emptyset$. 
To this end, let $\bar{\bm{z}} \ge \bm{0}$ be a feasible solution to \eqref{eq:computable_gale}. 
First, \eqref{eq:alter_gale_eq} holds since it is the same as \eqref{eq: computable_gale_eq} in this case. As for \eqref{eq:alter_gale_ineq}, we have
\begin{align}
[\bm{g}_D^T, \bm{g}_x^T, g_w, \bm{0}]\bar{\bm{z}} \le [\hat{\bm{g}}_D^T, \hat{\bm{g}}_x^T, g_w, \bm{0}]\bar{\bm{z}} < 0,
\end{align}
where the first inequality holds since $\bm{0}\le [\bm{g}_D^T, \bm{g}_x^T] \le [\hat{\bm{g}}_D^T, \hat{\bm{g}}_x^T]$ (element-wise inequality), while the second inequality holds since $\bar{\bm{z}}$ satisfies \eqref{eq:computable_gale}. Therefore, $\bar{\bm{z}}$ is also a feasible solution to \eqref{eq:alter_gale}, which verifies that $l \in F$.
\end{proof}

\subsubsection{Verification Algorithm}

All the elements in \eqref{eq:computable_gale} are known, and thus the existence of a solution can be checked by solving an LP. Based on this result, 
we propose Algorithm~\ref{alg: verify_LP} for verifying the estimated states of the remaining links, which iteratively updates $E_v$. {Each iteration of Algorithm~\ref{alg: verify_LP} involves solving $O(|E_H|)$ LPs, each of which has a time complexity that is polynomial\footnote{The exact order of the polynomial depends on the specific algorithm used to solve the LP~\cite{terlaky2013interior}. } in the number of decision variables ($|E_H|$) and the number of constraints ($|V_H|+|E_H|$)~\cite{terlaky2013interior}. Since Algorithm~\ref{alg: verify_LP} has at most $|E_H|$ iterations, the total time complexity of Algorithm~\ref{alg: verify_LP} is polynomial in $|E_H|$ and $|V_H|$.}

\begin{algorithm}\label{alg: verify_LP}
\SetAlgoLined
\SetKwFunction{Fmain}{FailEdgeDetection}
\SetKwInOut{Input}{input}\SetKwInOut{Output}{output}
\KwIn{$\tilde{\bfD}, \bfP, \bfDelta_{\bar{H}}, U_B, \eta,  E_H, E_v, \hat{F}, \hat{\bm{g}}_D, \hat{\bm{g}}_x$}
\While{$E_H\setminus E_v \ne \emptyset$}{
    $\bar{E}_v \leftarrow E_v$;\\
    \ForEach{$e\in E_H\setminus E_v$}{
        \If{$\exists \bm{z}\ge \bm{0}$ {satisfying} \eqref{eq:computable_gale} for $e$}{
        $\bar{E}_v \leftarrow \bar{E}_v \cup \{e\}$;\\
        Update $\hat{\bm{g}}_x$;
        }
    }
    \eIf{$|\bar{E}_v| > |E_v|$}{
    $E_v \leftarrow \bar{E}_v$;
    }{
    break\; 
    }
}
\caption{Verification with Partial Ground Truth}
\end{algorithm}

\subsection{Special Case of Connected Post-attack Grid}\label{sec: special_connected}

In this section, we study the special case that the grid is known to stay connected after the attack, which is assumed in most of the existing works \cite{Soltan18TCNS, Zhu12TPS, chen2014efficient}. In this case, FLD is modified by replacing constraints \eqref{eq:const_valid_start} and \eqref{eq:const_valid_load} with $\vDelta_H = \bm{0}$ (implied by the assumption of the connected post-attack grid). Next, we demonstrate how Algorithm~\ref{alg: Verification_Alg}-\ref{alg: verify_LP} will change in this case. To this end, we study the effect of $\vDelta_H = \bm{0}$ on Lemma~\ref{lem:ground_alter_gale}. Noting that according to \cite{Huang20arXiv}, any pair of $(\vDelta_H, \vx_H)$ satisfying \eqref{eq:pf_constraint} can be represented by $\vc\in \mathbb{R}^{|E_H|}$ as
\begin{align}
\vDelta_H = \vDelta_H^* + \tilde{\vD}_H\vc, \quad \vx_H  = \vx_H^* + \vI_{|E_H|}\vc.
\end{align}
Thus, we have $\tilde{\vD}_H\vc = \bm{0}$ due to $\vDelta_H = \vDelta_H^* = \bm{0}$, which is equivalent to requiring $\tilde{\vD}_H\vc \le \bm{0}$ and $-\tilde{\vD}_H\vc \le \bm{0}$.
Accordingly, $\vA_D$ and $\vg_D$ in \eqref{eq:alter_gale}, which used to model 
\eqref{eq:const_valid_start} and \eqref{eq:const_valid_load}, 
now become $\vA_D^T := [\tilde{\vD}_H^T, -\tilde{\vD}_H^T], \vg_D := \bm{0}$. 
The direct implication of $\vg_D = \bm{0}$ is that $f_{U,g} = \sum_{u\in U} f_{u,g} = 0, \forall U\subseteq V_H$. That is to say, Theorems~\ref{lem:certify_Ef_0}-\ref{lem:certify_Ef_2} still hold for the modified FLD except that $\hat{f}_{U,g} = \bm{0}$, which implies the following result:
\begin{corollary}\label{lem: No_mistake_1_cut}
If it is known that the post-attack grid $G' = (V,E\setminus F)$ is connected, then the state of any link that forms a 1-edge cut of $H$ will be identified correctly by a variation of FLD that replaces the constraints \eqref{eq:const_valid_start} and \eqref{eq:const_valid_load} by $\vDelta_H = \bm{0}$. 
\end{corollary}
\begin{proof}
As in the proof of Corollary~\ref{coro:verifiable condition of 1-edge cut}, for any link $e= (u_1, u_2)\in\hat{F}$ forming a cut of $H$, we can verify that $e\in F$ if $\min\{ {f}_{U_1,g}, {f}_{U_2,g} \} - \eta|\tilde{D}_{U_1,e}| <0$ (otherwise, $e$ must have been estimated as operational by  Theorem~\ref{lem:no_fa_hyper_direc}). Since ${f}_{U_i,g} = 0$ ($i=1,2$) if the grid remains connected after the attack and $|\tilde{D}_{U_1,e}|>0$ by Assumption~3, 
$e\in F$ can always be verified. Similar argument applies to any link $l\in E_H\setminus \hat{F}$.
\end{proof}

By Corollary~\ref{lem: No_mistake_1_cut}, the verification of the link states in $E_a$ 
can be skipped 
if the post-attack grid is known to stay connected.

\section{Performance Evaluation}\label{sec:Performance Evaluation}

We {first} test our solutions on the Polish power grid (``Polish system - winter 1999-2000 peak'') \cite{zimmerman2019matpower} 
with $2383$ nodes and $2886$ links, 
where parallel links are combined into one link. We generate the attacked area $H$ by randomly choosing one node as a starting point and performing a breadth first search to obtain $H$ with a predetermined $|V_H|$. 
We then randomly choose $|F|$ links within $H$ to fail.
{The generated $H$ consists of buses topologically close to each other, which will intuitively share communication links in connecting to the control center and can thus be blocked together once a cyber attack jams some of these links. Note, however, that our solution does not depend on this specific way of forming $H$.}
The phase angles of each island without any generator or load are set to $0$, and the rest are computed according to \eqref{eq:B theta = p}. %
For each setting of $|V_H|$ and $|F|$, we generate $300$ different $H$'s and $70$ different $F$'s per $H$. Each evaluated metric is shown via the mean and the $25^{\mbox{\small th}}$/$75^{\mbox{\small th}}$ percentile (indicated by the error bars). The threshold $\eta$ is set as $0.5$.\looseness=-1 


We first evaluate the fraction of verifiable links in $E_a$ (links in 1-edge cuts) and $E_c$ (links in 2-edge cuts, i.e., $E_c:=\bigcup_{s\in \mathcal{E}_c}s$), as shown in Fig.~\ref{fig:verifytopology_H40}. For each generated case (combination of $H$ and $F$), denote $E_{a,v}:=E_a\cap E_v$ and ${E}_{c,v}:={E}_{c}\cap E_v$. Then in Fig.~\ref{fig:verifytopology_H40}(a), we evaluate the fractions of testable and verifiable links in $E_a$ (${E}_{c}$) for failed links, i.e., $\frac{|E_{a}\cap F|}{|F|}$ ($\frac{|{E}_{c}\cap F|}{|F|}$) and $\frac{|E_{a,v}\cap F|}{|F|}$ ($\frac{|{E}_{c,v}\cap F|}{|F|}$).
The evaluation for operational links is conducted similarly in Fig.~\ref{fig:verifytopology_H40}(b). As can be seen, (\romannumeral1) the fractions of testable and verifiable links both stay almost constant with varying $|F|$, which demonstrates the robustness of Algorithm~\ref{alg: Verification_Alg}; (\romannumeral2) among the testable links ($E_a\cup {E}_c$), most of the failed links are verifiable, but only half of the operational links are verifiable; (\romannumeral3) compared to links in ${E}_{c}$, links in $E_a$ have a higher chance of being verifiable, which
indicates that it is easier to recover the states of the critical links in the attacked area (that form 1-edge cuts).

\begin{figure}
\begin{minipage}{.495\linewidth}\label{mini:topology_Ef}
  \centerline{
  \includegraphics[width=1\columnwidth]{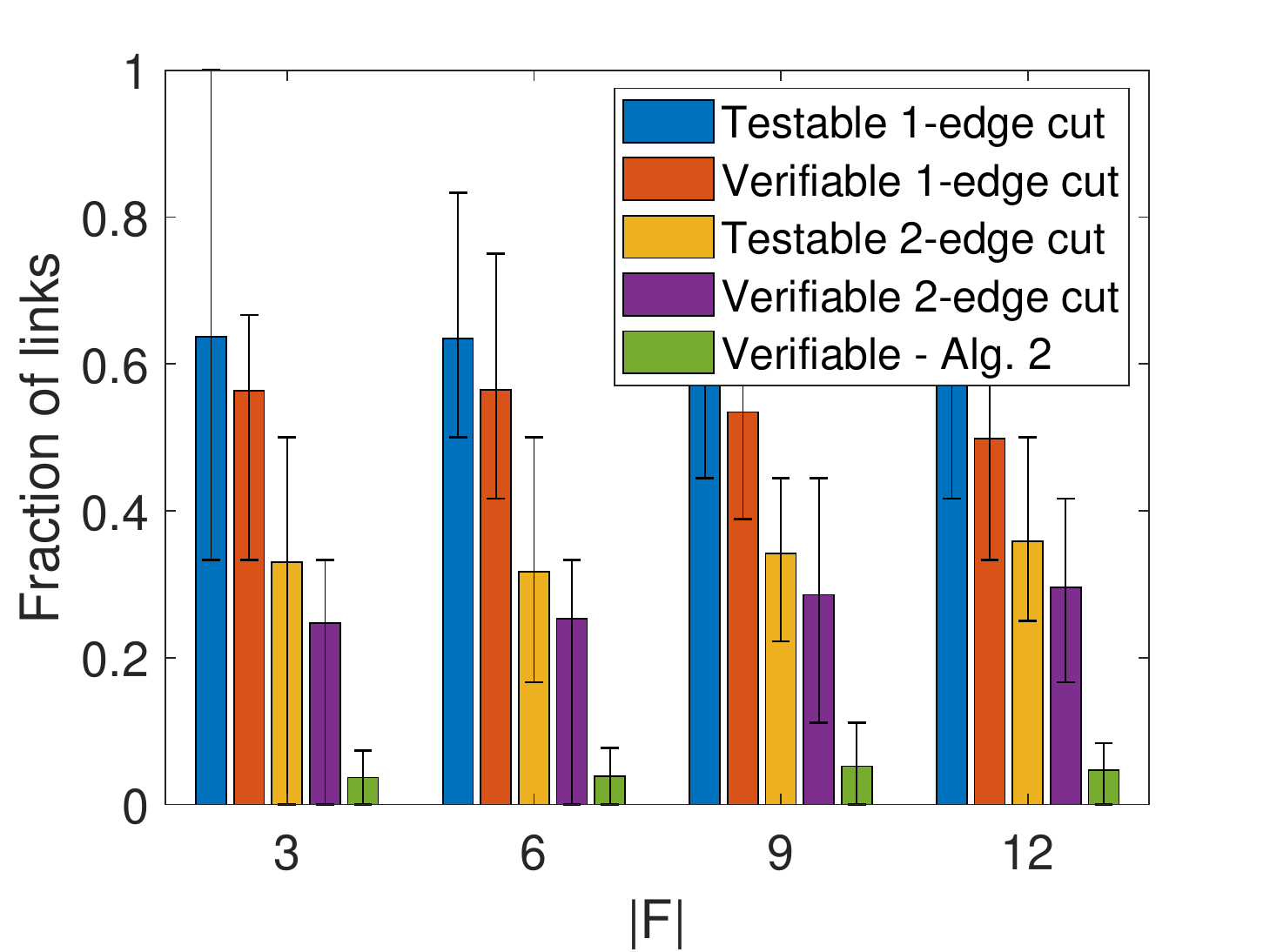}}
  \vspace{-.05em}
  \centerline{\small (a) Fraction of failed links}
\end{minipage}\hfill
\begin{minipage}{.495\linewidth}\label{mini:topology_E2}
 \centerline{
  \includegraphics[width=1\columnwidth]{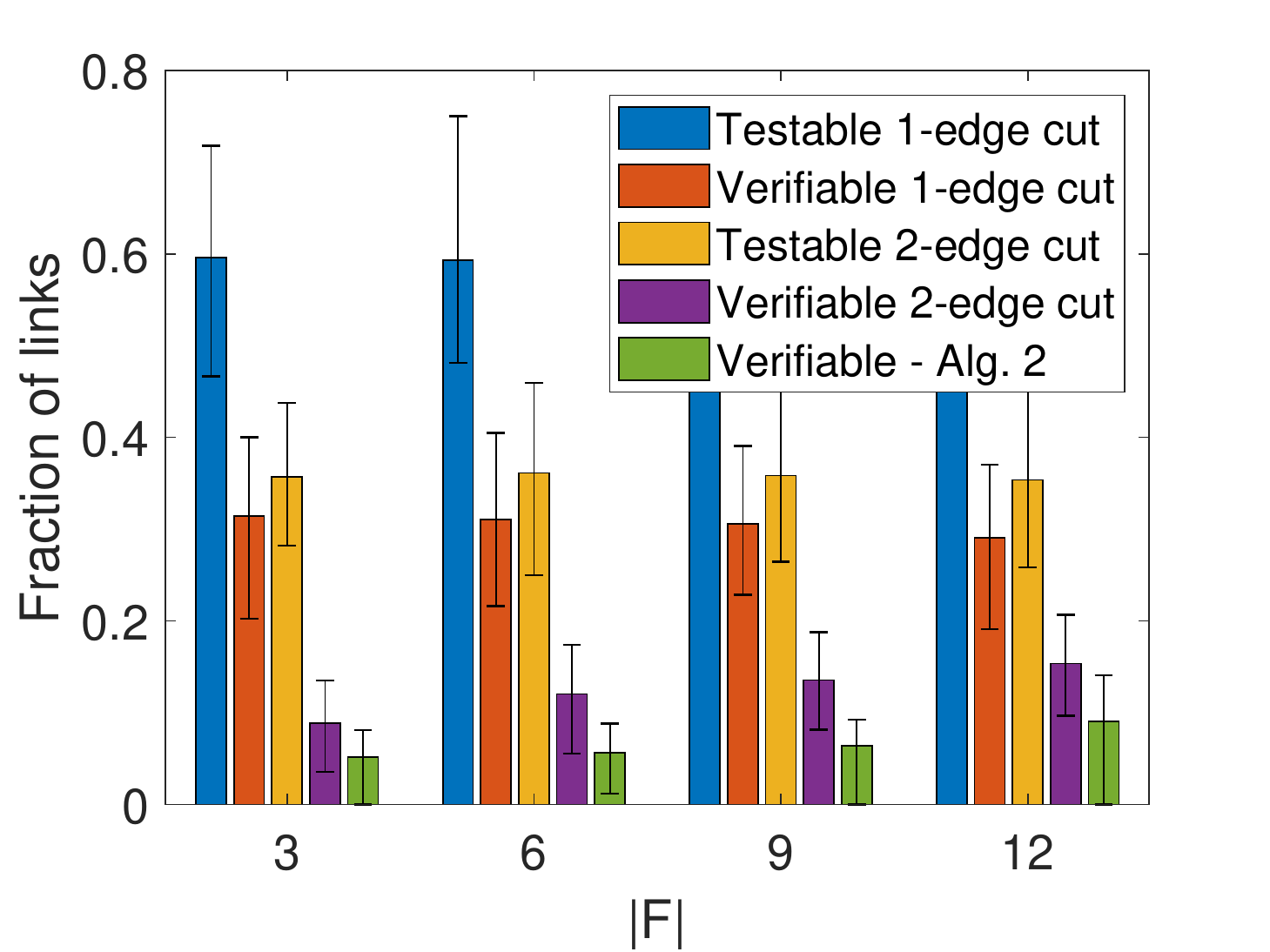}}
  \vspace{-.05em}
  \centerline{\small (b) Fraction of operational links}
\end{minipage}
  \caption{Fraction of testable/verifiable links in Polish system 
  ($|V_H|=40$).}
  \label{fig:verifytopology_H40}
    \vspace{-1em}
\end{figure}

\begin{table}[tb]
\footnotesize
\renewcommand{\arraystretch}{1.3}
\caption{Percentage of cases that Algorithm~\ref{alg: verify_LP} verifies additional links in Polish system} \label{tab:frac_alg2_effective}
\centering
\begin{tabular}{c|c|c|c|c}
  \hline
  Type of links & $|F| = 3$& $|F| = 6$& $|F| = 9$& $|F| = 12$  \\
  \hline
 Failed Links & $18.86\%$ & $31.94\%$ & $45.69\%$ & $54.42\%$ \\
  \hline
 Operational Links & $81.13\%$ & $84.24\%$ & $85.41\%$ & $85.69\%$ \\
  \hline
 All Links & $83.75\%$ & $88.23\%$ & $91.02\%$ & $91.48\%$ \\
  \hline
\end{tabular}
\end{table}

Next, we evaluate two metrics to study the value of Algorithm~\ref{alg: verify_LP}. The first is the fraction of links verified by Algorithm~\ref{alg: verify_LP} but not Algorithm~\ref{alg: Verification_Alg}, as shown in Fig.~\ref{fig:verifytopology_H40} as 'Verifiable - Alg.~2'. The second is the {percentage of cases} that Algorithm~\ref{alg: verify_LP} can verify additional links, given in Table~\ref{tab:frac_alg2_effective} for different $|F|$. We observe that Algorithm~\ref{alg: verify_LP} can usually verify more links based on the results of Algorithm~\ref{alg: Verification_Alg}, although the number of additionally verified links is not large.

\begin{figure}
\begin{minipage}{.495\linewidth}\label{subfig:Verify_EF_H40}
  \centerline{
  \includegraphics[width=1\columnwidth]{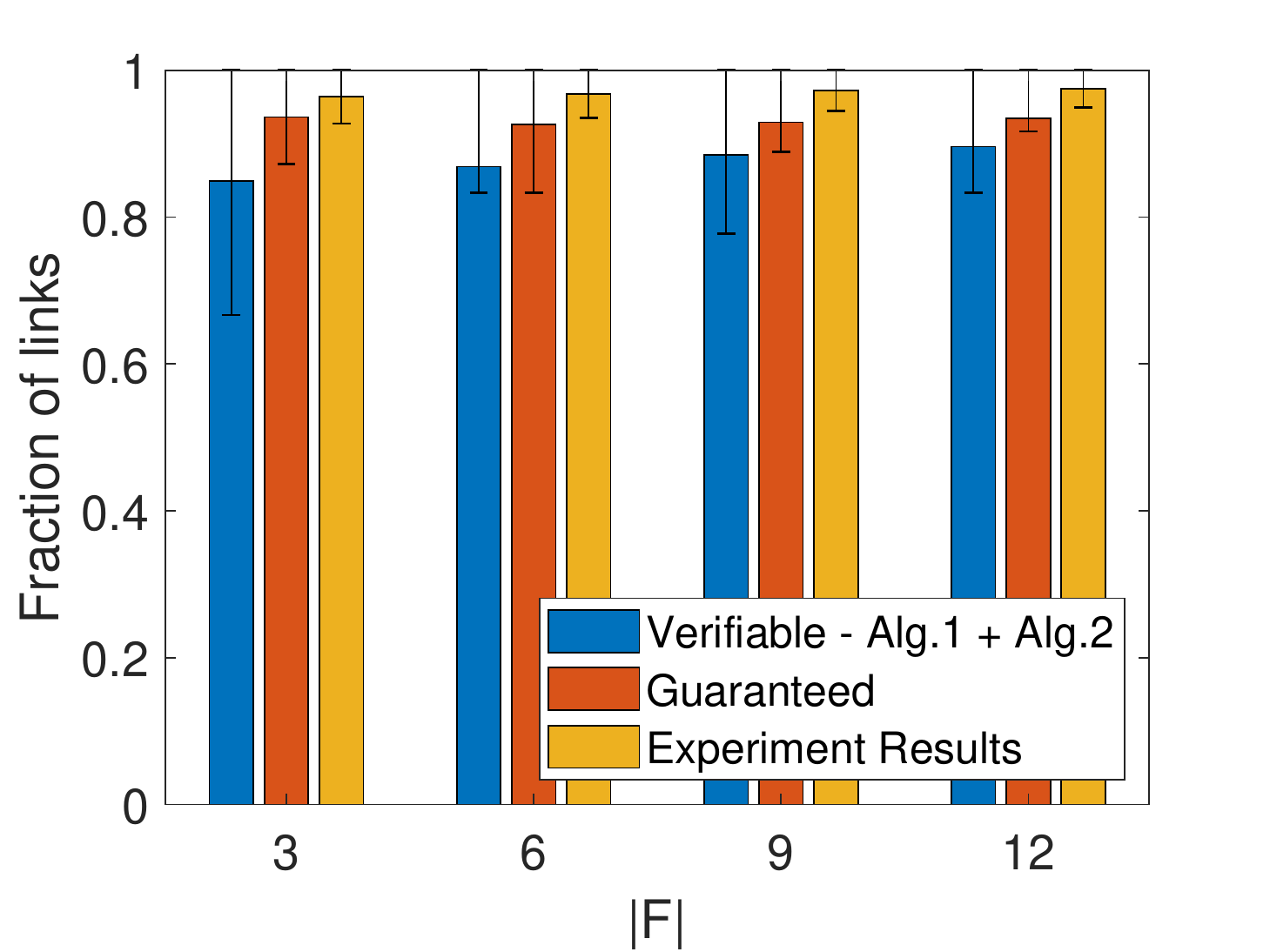}}
  \centerline{\small (a) Fraction of failed links.}
\end{minipage}\hfill
\begin{minipage}{.495\linewidth}\label{subfig:Verify_E2_H40}
 \centerline{
  \includegraphics[width=1\columnwidth]{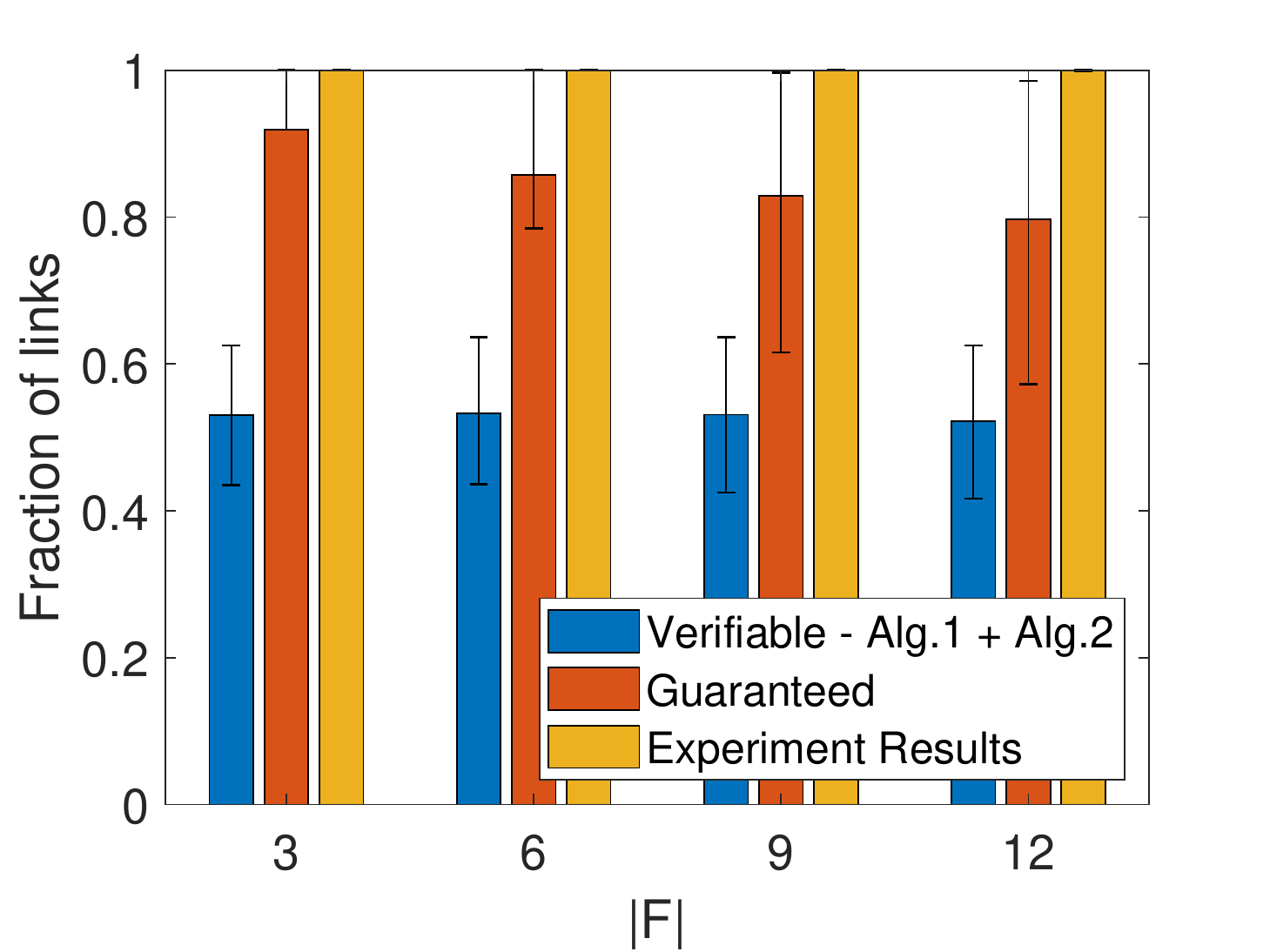}}
  \centerline{\small (b) Fraction of operational links.}
\end{minipage}
  \caption{Comparison between verifiable links, theoretically guaranteed links, and actually correctly identified links in Polish system ($|V_H|=40$).}
  \label{fig:verify_guarantee_general}
    \vspace{-1em}
\end{figure}

Then, we compare the fraction of \emph{verifiable links} with unknown ground truth {of $F$} to the fraction of links whose states are \emph{guaranteed} to be correctly estimated by FLD based on the ground truth $F$ according to Lemma~\ref{lem:ground_alter_gale} (`Guaranteed') and the \emph{actual} fraction of links whose states are correctly estimated by FLD (`Experiment Results'), as shown in Fig.~\ref{fig:verify_guarantee_general}. We see that most of the failed links are verifiable, while only half of the operational links are verifiable. This indicates that most (more than $90\%$) of the unverifiable links are operational. To understand such a phenomenon, we observe in experiments that many operational links carry small post-attack power flow, which makes the conditions in Theorem~\ref{lem:certify_Ef_0}-\ref{lem:certify_Ef_2} hard to satisfy. On the contrary, the values of hypothetical power flows on failed links are usually large. {Nevertheless, the fraction of links whose states are correctly identified by FLD is much higher: out of all the failed links, over $80\%$ will be estimated as failed and verified as so, while another $15\%$ will be estimated as failed but not verified; out of all the operational links, over $50\%$ will be estimated and verified as operational, while the rest will also be estimated as operational but not verified.}
\looseness=0


Finally, for the special case that the post-attack grid stays connected, we study the benefits of knowing the connectivity and the corresponding modification in Section~\ref{sec: special_connected}, as shown in Fig.~\ref{fig:verifyGuaranteeConnected_H40} and Table~\ref{tab:frac_connected_grid}. Specifically, `X-agnostic' denotes the performance of `X' without knowing the connectivity, while `X-known' denotes the counterpart that adopts the modification in Section~\ref{sec: special_connected}. The meaning of `X' is the same as in  Fig.~\ref{fig:verify_guarantee_general}. In Table~\ref{tab:frac_connected_grid}, we evaluate the percentage of randomly generated cases ($H$ and $F$) that the post-attack grid $G'$ remains connected. We observe that (\romannumeral1) the knowledge of connectivity can help verify more than $10\%$ additional failed links and $30\%$ additional operational links; (\romannumeral2) when $|F|$ is small (e.g., $|F|\le 3$), $G'$ remains connected in the majority of the cases. 
These results indicate the value of the knowledge of connectivity.


\begin{table}[tb]
\footnotesize
\renewcommand{\arraystretch}{1.3}
\caption{Percentage of cases of connected post-attack Polish system ($|V_H| = 40$)} \label{tab:frac_connected_grid}
\centering
\begin{tabular}{c|c|c|c}
  \hline
  $|F| = 3$& $|F| = 6$& $|F| = 9$& $|F| = 12$  \\
  \hline
  57.12$\%$& 26.33$\%$&  11.87$\%$ & 5.04$\%$  \\
  \hline
\end{tabular}
\end{table}

{To validate our observations, we further evaluate our solutions on the IEEE 300-bus system extracted from MATPOWER \cite{zimmerman2019matpower}, as shown in Fig.~\ref{fig:IEEE_verify_guarantee_general}--\ref{fig:IEEE_verify_guarantee_connected} and Table~\ref{tab:IEEE_frac_connected_grid}. The configuration of these experiments is the same as before, except that $|V_H| = 20$ due to the smaller scale of the test system. Compared with Fig.~\ref{fig:verify_guarantee_general}--\ref{fig:verifyGuaranteeConnected_H40} and Table~\ref{tab:frac_connected_grid}, all the results from the IEEE 300-bus system are qualitatively similar to those from the Polish system, and hence validate the generality of our previously observations.  }

\begin{figure}
\begin{minipage}{.495\linewidth}\label{mini:guarantee_ef}
  \centerline{
  \includegraphics[width=1\columnwidth]{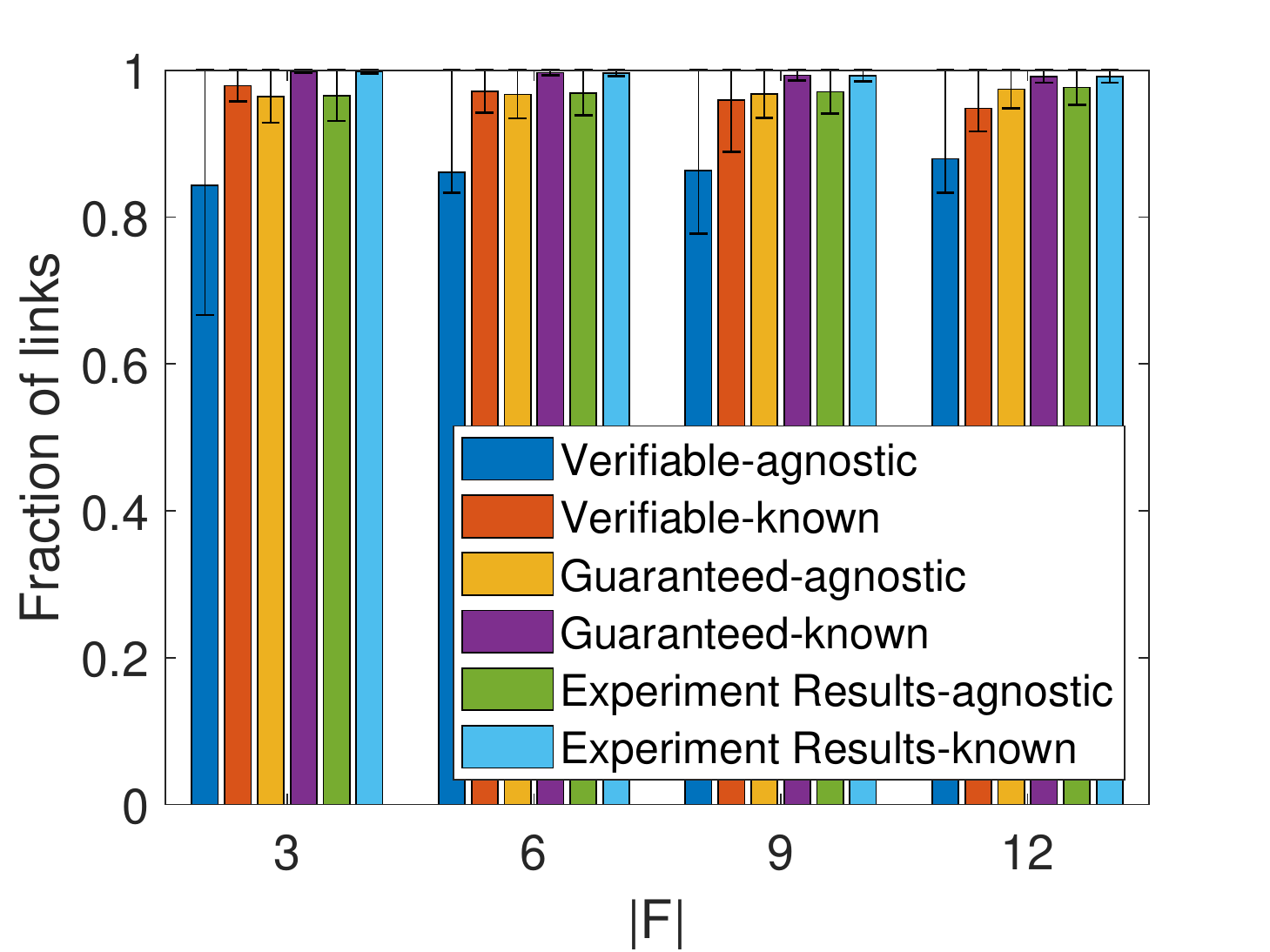}}
  \centerline{\small (a) Fraction of failed links.}
\end{minipage}\hfill
\begin{minipage}{.495\linewidth}
 \centerline{
  \includegraphics[width=1\columnwidth]{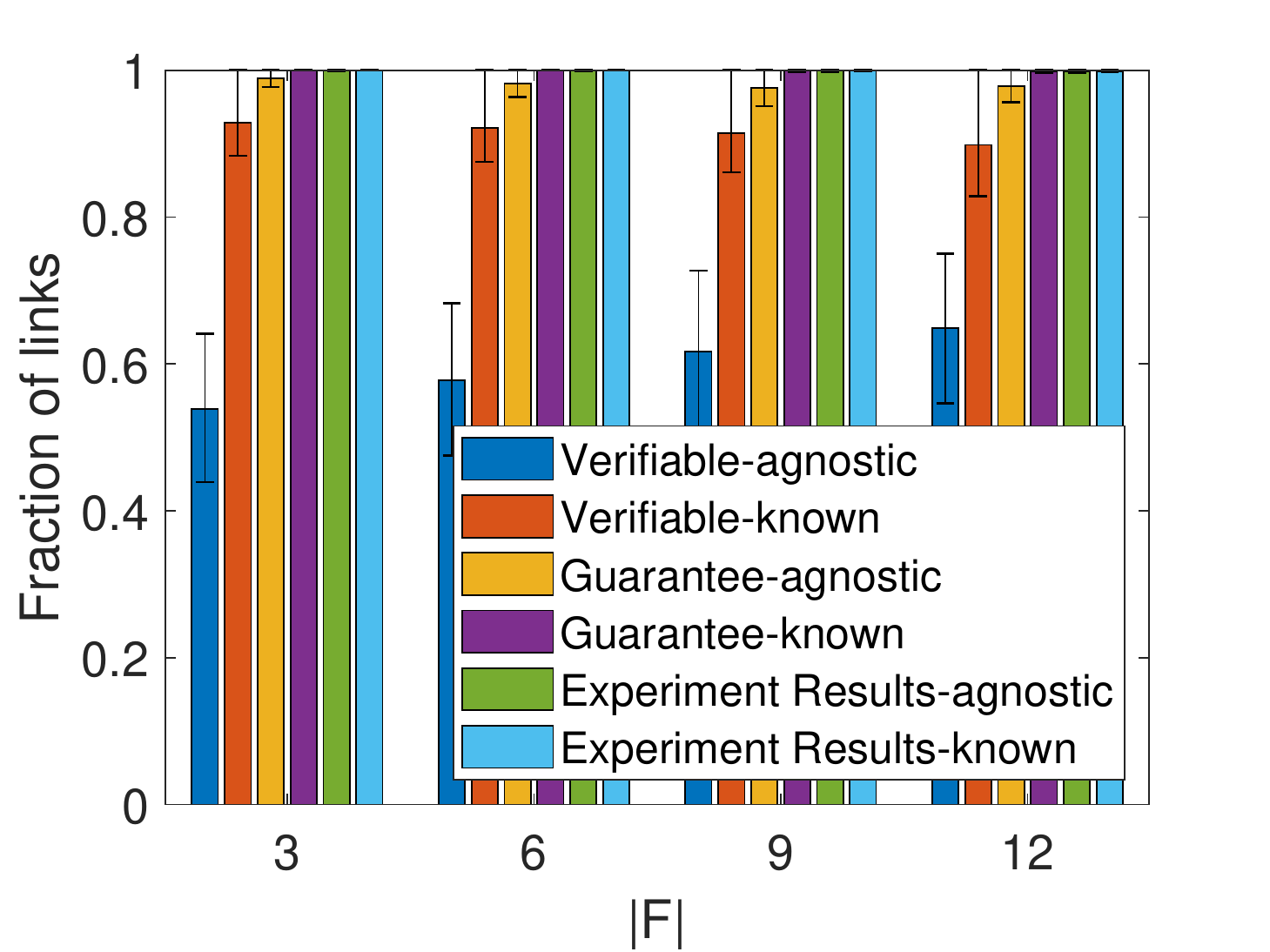}}
  \centerline{\small (b) Fraction of operational links.}
\end{minipage}
  \caption{Performance comparison for connected post-attack Polish system ($|V_H|=40$).}
  \label{fig:verifyGuaranteeConnected_H40}
    \vspace{-1em}
\end{figure}

\begin{figure}
\begin{minipage}{.495\linewidth}\label{subfig:IEEE_Verify_EF_H40}
  \centerline{
  \includegraphics[width=1\columnwidth]{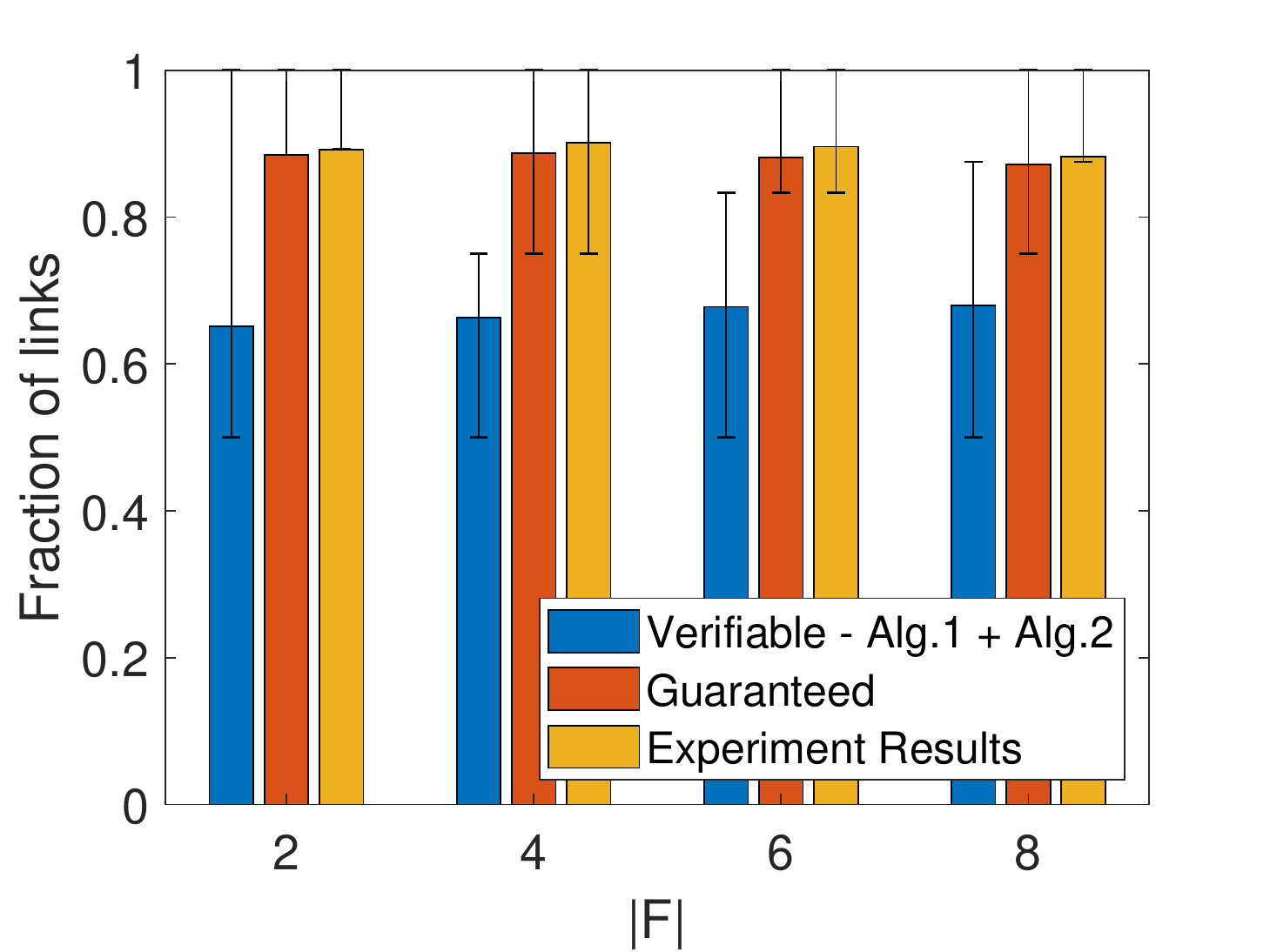}}
  \centerline{\small (a) Fraction of failed links.}
\end{minipage}\hfill
\begin{minipage}{.495\linewidth}\label{subfig:IEEE_Verify_E2_H40}
 \centerline{
  \includegraphics[width=1\columnwidth]{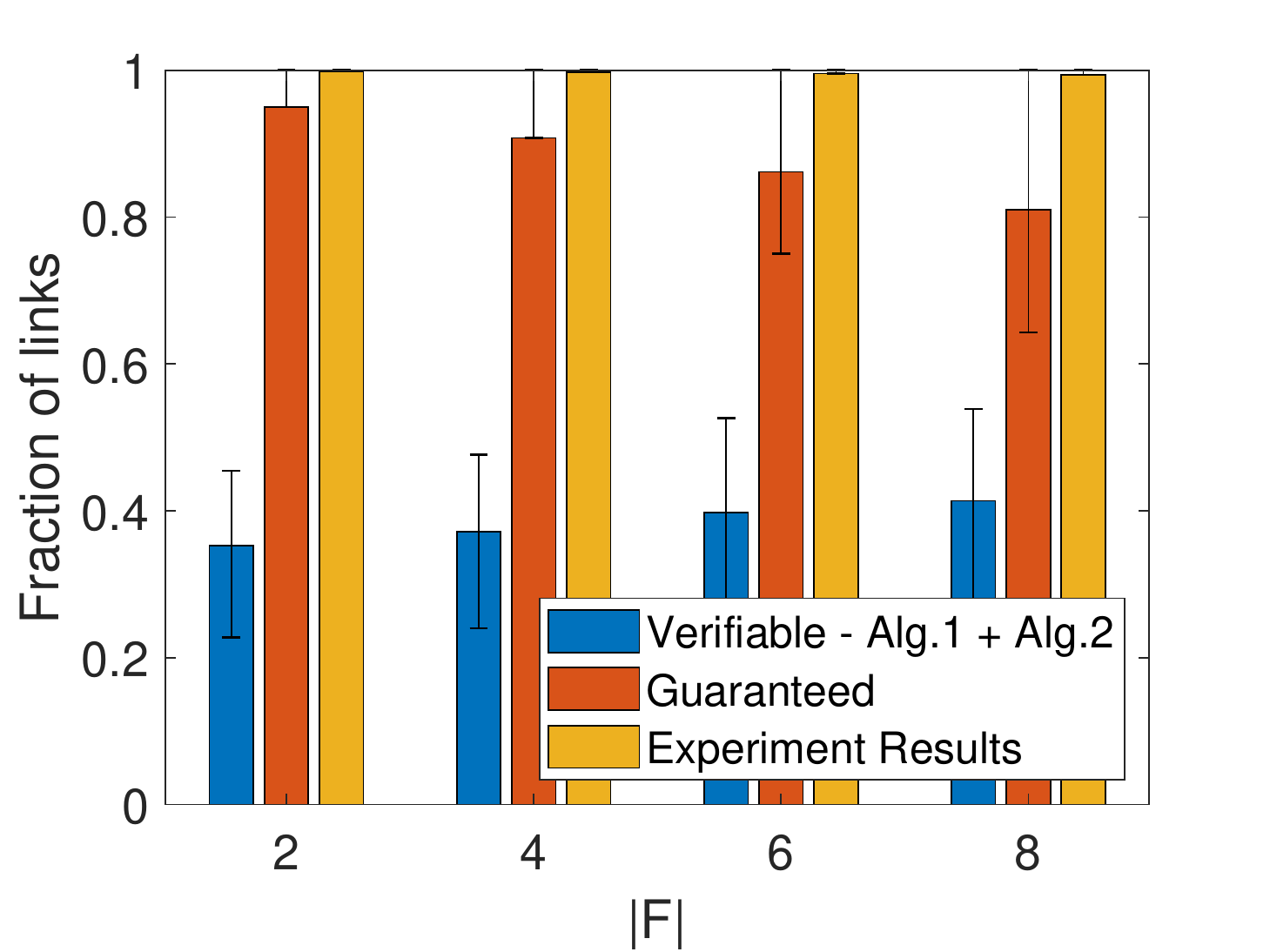}}
  \centerline{\small (b) Fraction of operational links.}
\end{minipage}
  {\caption{Comparison between verifiable links, theoretically guaranteed links, and actually correctly identified links in IEEE 300-bus system ($|V_H|=20$).}
  \label{fig:IEEE_verify_guarantee_general}}
\end{figure}

\begin{figure}
\begin{minipage}{.495\linewidth}\label{subfig:IEEE_VerifyConn_EF_H40}
  \centerline{
  \includegraphics[width=1\columnwidth]{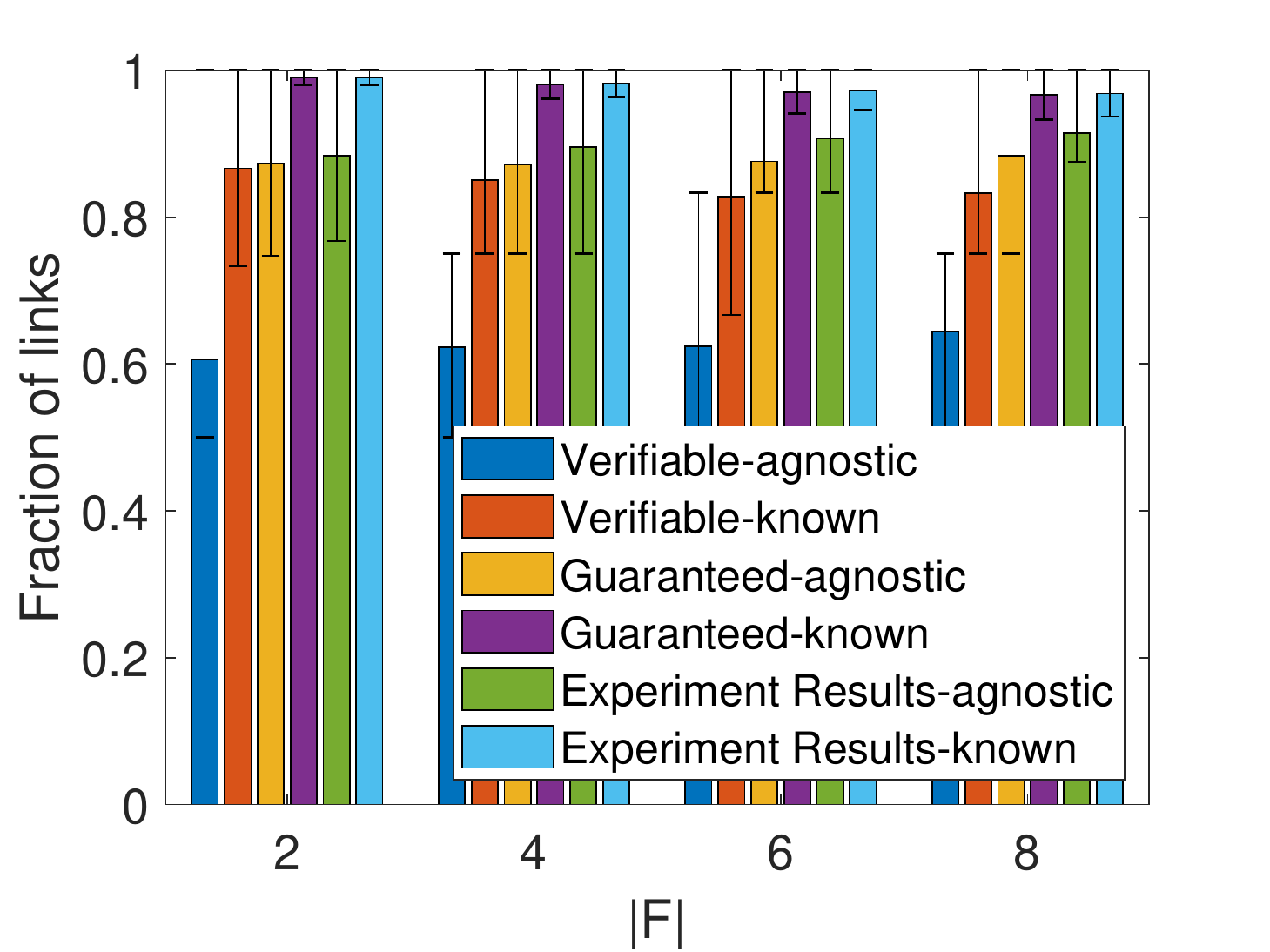}}
  \centerline{\small (a) Fraction of failed links.}
\end{minipage}\hfill
\begin{minipage}{.495\linewidth}\label{subfig:IEEE_VerifyConn_E2_H40}
 \centerline{
  \includegraphics[width=1\columnwidth]{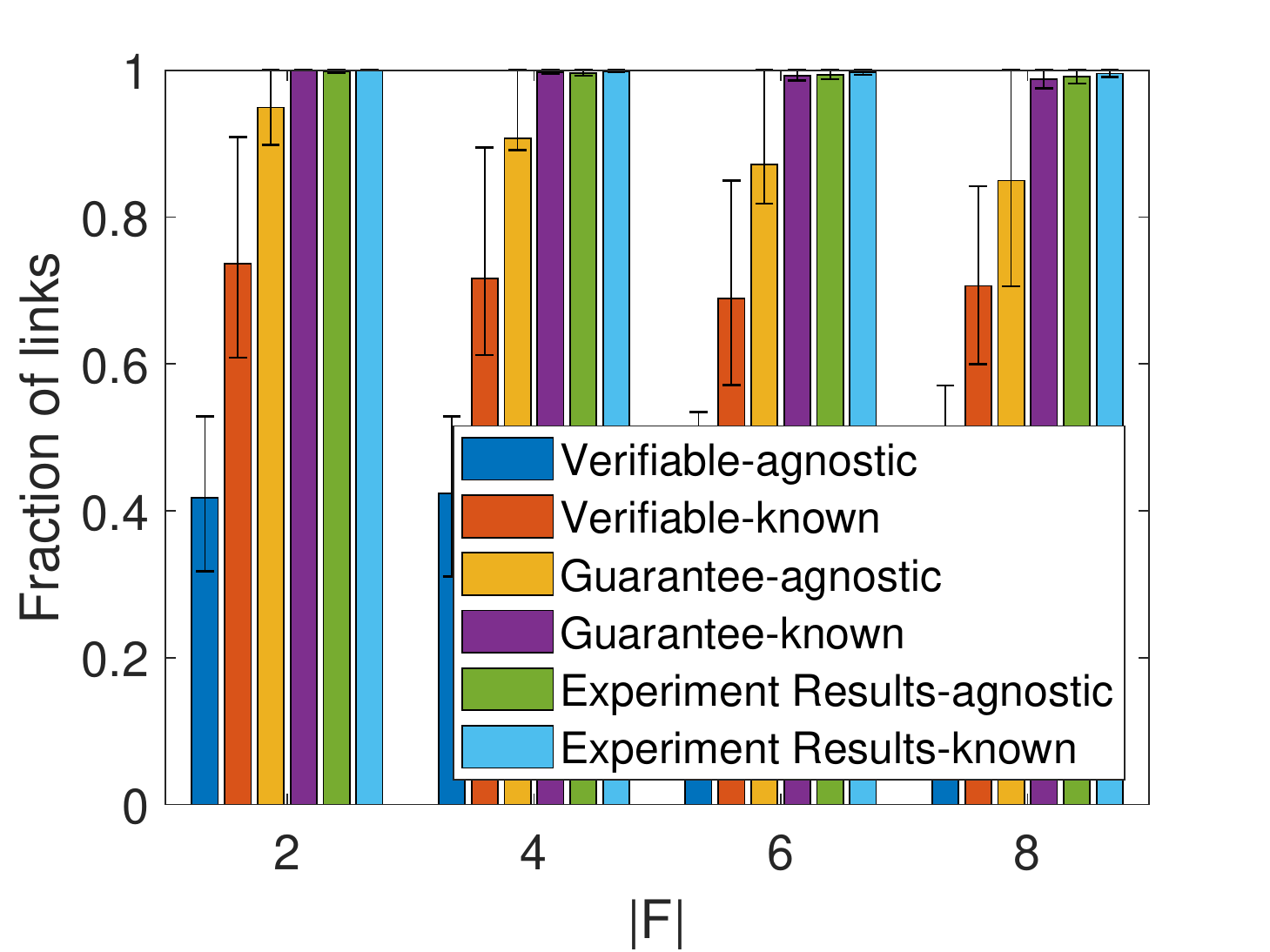}}
  \centerline{\small (b) Fraction of operational links.}
\end{minipage}
  {\caption{Performance comparison for connected post-attack IEEE 300-bus system ($|V_H|=20$). }
  \label{fig:IEEE_verify_guarantee_connected}}
\end{figure}

\begin{table}[tb]
\footnotesize
\renewcommand{\arraystretch}{1.3}
{\caption{Percentage of cases of connected post-attack IEEE 300-bus system ($|V_H|=20$)} \label{tab:IEEE_frac_connected_grid}}
\centering
\begin{tabular}{c|c|c|c}
  \hline
  $|F| = 2$& $|F| = 4$& $|F| = 6$& $|F| = 8$  \\
  \hline
  73.73$\%$& 51.10$\%$&  32.89$\%$ & 18.54$\%$  \\
  \hline
\end{tabular}
\end{table}


\section{Conclusion}\label{sec:Conclusion}

We considered the problem of localizing failed links in a smart grid under a cyber-physical attack that blocks sensor data from the attacked area and disconnects an unknown subset of links within this area that may disconnect the grid. Building on top of a recently proposed failure detection algorithm (FLD) that has shown empirical success, we focused on verifying the correctness of the estimated link states, by developing theoretical conditions that can be verified based on observable information and polynomial-time algorithms that use these conditions to verify link states. Our evaluations based on the Polish power grid showed that the proposed algorithms are highly successful in verifying the states of truly failed links. Compared to the previous solutions {(including \cite{Huang20arXiv})} for link state estimation that label links with binary states (failed/operational) without guaranteed correctness, our solution labels links with ternary states (failed/operational/unverifiable), where the states of verifiable links are identified with guaranteed correctness. This, together with the observation that most of the unverifiable links are operational, provides valuable information for planning repairs during the recovery process.


\bibliographystyle{IEEEtran}
\bibliography{myBib}

\begin{thebibliography}{10}
\providecommand{\url}[1]{#1}
\csname url@samestyle\endcsname
\providecommand{\newblock}{\relax}
\providecommand{\bibinfo}[2]{#2}
\providecommand{\BIBentrySTDinterwordspacing}{\spaceskip=0pt\relax}
\providecommand{\BIBentryALTinterwordstretchfactor}{4}
\providecommand{\BIBentryALTinterwordspacing}{\spaceskip=\fontdimen2\font plus
\BIBentryALTinterwordstretchfactor\fontdimen3\font minus
  \fontdimen4\font\relax}
\providecommand{\BIBforeignlanguage}[2]{{%
\expandafter\ifx\csname l@#1\endcsname\relax
\typeout{** WARNING: IEEEtran.bst: No hyphenation pattern has been}%
\typeout{** loaded for the language `#1'. Using the pattern for}%
\typeout{** the default language instead.}%
\else
\language=\csname l@#1\endcsname
\fi
#2}}
\providecommand{\BIBdecl}{\relax}
\BIBdecl

\bibitem{UkraineAttack}
``Analysis of the cyber attack on the {Ukrainian} power grid,'' March 2016,
  \url{https://ics.sans.org/media/E-ISAC\_SANS\_Ukraine\_DUC\_5.pdf}.

\bibitem{Soltan18TCNS}
S.~Soltan, M.~Yannakakis, and G.~Zussman, ``Power grid state estimation
  following a joint cyber and physical attack,'' \emph{IEEE Transactions on
  Control of Network Systems}, vol.~5, no.~1, pp. 499--512, 2018.

\bibitem{Soltan17PES}
S.~Soltan and G.~Zussman, ``Power grid state estimation after a cyber-physical
  attack under the {AC} power flow model,'' in \emph{IEEE PES-GM}, 2017.

\bibitem{yudi20SmartGridComm}
Y.~Huang, T.~He, N.~R. Chaudhuri, and T.~L. Porta, ``Power grid state
  estimation under general cyber-physical attacks,'' in \emph{IEEE
  SmartGridComm}.\hskip 1em plus 0.5em minus 0.4em\relax IEEE, 2020.

\bibitem{huang2012state}
Y.-F. Huang, S.~Werner, J.~Huang, N.~Kashyap, and V.~Gupta, ``State estimation
  in electric power grids: Meeting new challenges presented by the requirements
  of the future grid,'' \emph{IEEE Signal Processing Magazine}, vol.~29, no.~5,
  pp. 33--43, 2012.

\bibitem{tate2008line}
J.~E. Tate and T.~J. Overbye, ``Line outage detection using phasor angle
  measurements,'' \emph{IEEE Transactions on Power Systems}, vol.~23, no.~4,
  pp. 1644--1652, 2008.

\bibitem{tate2009double}
------, ``Double line outage detection using phasor angle measurements,'' in
  \emph{2009 IEEE Power \& Energy Society General Meeting}.\hskip 1em plus
  0.5em minus 0.4em\relax IEEE, 2009, pp. 1--5.

\bibitem{Zhu12TPS}
H.~Zhu and G.~B. Giannakis, ``Sparse overcomplete representations for efficient
  identification of power line outages,'' \emph{IEEE Transactions on Power
  Systems}, vol.~27, no.~4, pp. 2215--2224, November 2012.

\bibitem{chen2014efficient}
J.-C. Chen, W.-T. Li, C.-K. Wen, J.-H. Teng, and P.~Ting, ``Efficient
  identification method for power line outages in the smart power grid,''
  \emph{IEEE Transactions on Power Systems}, vol.~29, no.~4, pp. 1788--1800,
  2014.

\bibitem{garcia2015line}
M.~Garcia, T.~Catanach, S.~Vander~Wiel, R.~Bent, and E.~Lawrence, ``Line outage
  localization using phasor measurement data in transient state,'' \emph{IEEE
  Transactions on Power Systems}, vol.~31, no.~4, pp. 3019--3027, 2015.

\bibitem{zhao2019learning}
Y.~Zhao, J.~Chen, and H.~V. Poor, ``A learning-to-infer method for real-time
  power grid multi-line outage identification,'' \emph{IEEE Transactions on
  Smart Grid}, vol.~11, no.~1, pp. 555--564, 2020.

\bibitem{soltan2018react}
S.~Soltan, M.~Yannakakis, and G.~Zussman, ``React to cyber attacks on power
  grids,'' \emph{IEEE Transactions on Network Science and Engineering}, vol.~6,
  no.~3, pp. 459--473, 2018.

\bibitem{soltan2018expose}
S.~Soltan and G.~Zussman, ``Expose the line failures following a cyber-physical
  attack on the power grid,'' \emph{IEEE Transactions on Control of Network
  Systems}, vol.~6, no.~1, pp. 451--461, 2018.

\bibitem{Adu01TPD}
T.~Adu, ``A new transmission line fault locating system,'' \emph{IEEE
  Transactions on Power Delivery}, vol.~16, no.~4, pp. 498--503, October 2001.

\bibitem{salim2009extended}
R.~H. Salim, M.~Resener, A.~D. Filomena, K.~R.~C. De~Oliveira, and A.~S.
  Bretas, ``Extended fault-location formulation for power distribution
  systems,'' \emph{IEEE Transactions on Power Delivery}, vol.~24, no.~2, pp.
  508--516, 2009.

\bibitem{Codino17TEC}
A.~Codino, Z.~Wang, R.~Razzaghi, M.~Paolone, and F.~Rachidi, ``An alternative
  method for locating faults in transmission line networks based on time
  reversal,'' \emph{IEEE Transactions on Electromagnetic Compatibility},
  vol.~59, no.~5, pp. 1601--1612, October 2017.

\bibitem{saha2009fault}
M.~M. Saha, J.~J. Izykowski, and E.~Rosolowski, \emph{Fault location on power
  networks}.\hskip 1em plus 0.5em minus 0.4em\relax Springer Science \&
  Business Media, 2009.

\bibitem{Huang20arXiv}
Y.~Huang, T.~He, N.~R. Chaudhuri, and T.~L. Porta, ``Power grid state
  estimation under general cyber-physical attacks,'' arXiv: 2009.02377, Sept.
  2020, \url{https://arxiv.org/abs/2009.02377}.

\bibitem{WASA}
``{WASA} and the roadmap to {WAMPAC} at {SDG\&E},'' September 2020,
  \url{https://quanta-technology.com/wp-content/uploads/2020/09/WASA-and-the-Roadmap-to-WAMPAC-at-SDGE.pdf}.

\bibitem{PMUdeployment}
``{SynchroPhasor} technology fact sheet,'' North American SynchroPhasor
  Initiative, October 2014,
  \url{https://www.naspi.org/sites/default/files/reference\_documents/33.pdf}.

\bibitem{dagle2010north}
J.~E. Dagle, ``{The North American SynchroPhasor Initiative (NASPI)},'' in
  \emph{IEEE PES General Meeting}.\hskip 1em plus 0.5em minus 0.4em\relax IEEE,
  2010, pp. 1--3.

\bibitem{jones2013three}
K.~D. Jones, J.~S. Thorp, and R.~M. Gardner, ``Three-phase linear state
  estimation using phasor measurements,'' in \emph{2013 IEEE Power \& Energy
  Society General Meeting}.\hskip 1em plus 0.5em minus 0.4em\relax IEEE, 2013,
  pp. 1--5.

\bibitem{jones2014methodology}
K.~D. Jones, A.~Pal, and J.~S. Thorp, ``Methodology for performing
  synchrophasor data conditioning and validation,'' \emph{IEEE Transactions on
  Power Systems}, vol.~30, no.~3, pp. 1121--1130, 2014.

\bibitem{WAMPACsecurity}
``Wide area monitoring, protection, and control systems {(WAMPAC)} standards
  for cyber security requirements,'' National Electric Sector Cybersecurity
  Organization Resource (NESCOR), October 2012,
  \url{https://smartgrid.epri.com/doc/ESRFSD.pdf}.

\bibitem{pal2006robust}
B.~Pal and B.~Chaudhuri, \emph{Robust control in power systems}.\hskip 1em plus
  0.5em minus 0.4em\relax Springer Science \& Business Media, 2006.

\bibitem{lu2016under}
M.~Lu, W.~ZainalAbidin, T.~Masri, D.~Lee, and S.~Chen, ``Under-frequency load
  shedding (ufls) schemes-a survey,'' \emph{International Journal of Applied
  Engineering Research}, vol.~11, no.~1, pp. 456--472, 2016.

\bibitem{terlaky2013interior}
T.~Terlaky, \emph{Interior point methods of mathematical programming}.\hskip
  1em plus 0.5em minus 0.4em\relax Springer Science \& Business Media, 2013,
  vol.~5.

\bibitem{zimmerman2019matpower}
R.~D. Zimmerman and C.~E. Murillo-S{\'a}nchez, ``Matpower 7.0 user’s
  manual,'' \emph{Power Systems Engineering Research Center}, vol.~9, 2019.

\end{thebibliography}

\end{document}